\documentclass[10pt,a4paper]{article}

\usepackage{amsmath,amsgen,latexsym}
\usepackage{amstext,amssymb,amsfonts,latexsym}
\usepackage{theorem}
\usepackage{pifont}
\usepackage[dvips]{graphics,epsfig}
\usepackage[dvips]{graphicx}

 \newcommand{\bs}{\bigskip}
 
 \newcommand{\n}{\noindent}
 \newcommand{\s}{\smallskip}
 \newcommand{\hs}[1]{\hspace*{ #1 mm}}
 \newcommand{\vs}[1]{\vspace*{ #1 mm}}

 \newcommand{\setempty}{\mathrm{\O}}
 \newcommand{\real}{\mathbb{R}}
 \newcommand{\nat}{\mathbb{N}}
 
 \newcommand{\integer}{\mathbb{Z}}
 \newcommand{\rational}{\mathbb{Q}}
 \newcommand{\complex}{\mathbb{C}}
 \newcommand{\algebraic}{\mathbb{A}}

 \newcommand{\ie}{\textrm{i.e.},\hspace*{2mm}}

 \newcommand{\etalc}{\textrm{et al.}}

 \newcommand{\BB}{{\cal B}}
 
 \newcommand{\FF}{{\cal F}}

 \newcommand{\HH}{{\cal H}}

 \newcommand{\GG}{{\cal G}}
 \newcommand{\MM}{{\cal M}}

 \newcommand{\UU}{{\cal U}}

 \newcommand{\fp}{\mathrm{FP}}

 \newcommand{\comb}[2]{\left(\:\begin{subarray}{c} #1 \\%
      #2 \end{subarray}\right)}

 \def\bbox{\vrule height6pt width6pt depth1pt}

\theoremstyle{plain}
\theoremheaderfont{\bfseries}
 \newtheorem{theorem}{Theorem}[section]
 \newtheorem{lemma}[theorem]{Lemma}
 \newtheorem{proposition}[theorem]{Proposition}

 {\theorembodyfont{\rmfamily}
  \newtheorem{definition}[theorem]{Definition}}
 {\theorembodyfont{\rmfamily} }
 {\theorembodyfont{\rmfamily} }

 \newtheorem{claim}{Claim}

 \newenvironment{proof}{\par \noindent
            {\bf Proof. \hs{2}}}{\hfill$\Box$ \vspace*{3mm}}

 \newenvironment{proofof}[1]{\vspace*{5mm} \par \noindent
         {\bf Proof of #1.\hs{2}}}{\hfill$\Box$ \vspace*{3mm}}

 \newcommand{\pair}[1]{\langle #1 \rangle}

\newcommand{\ignore}[1]{}

\newcommand{\sharpcsp}{\#\mathrm{CSP}}

\newcommand{\APreduces}{\leq_{\mathrm{AP}}}
\newcommand{\APequiv}{\equiv_{\mathrm{AP}}}

\newcommand{\DG}{{\cal DG}}

\newcommand{\ED}{{\cal ED}}

\newcommand{\NAND}{{\cal NAND}}

\newcommand{\OR}{{\cal OR}}
\newcommand{\csp}{csp}
\newcommand{\AZ}{{\cal AZ}}

\newcommand{\econst}{\leq_{e\mbox{-}con}}

\begin{document}
%%%%%%%%%%%%%%%%%%
%%%%%%%%%%%%%%%%%%
\pagestyle{plain}
\setcounter{page}{1}

\begin{center}
{\Large {\bf Constant Unary Constraints and
Symmetric Real-Weighted \s\\ Counting Constraint Satisfaction Problems}}\footnote{A preliminary version under a slightly concise title appeared in the Proceedings of the 23rd International Symposium on Algorithms and Computation (ISAAC 2012),  Taipei, Taiwan, December 19--21, 2012, Lecture Notes in Computer Science, Springer-Verlag, vol. 7676, pp. 237--246, 2012.} \bs\\
{\sc Tomoyuki Yamakami}\footnote{Present Affiliation: Department of Information Science, University of Fukui, 3-9-1 Bunkyo, Fukui 910-8507, Japan} \bs\\
\end{center}

%%%%%%%%%%%%%%%%%%
\begin{quote}
\n{\bf Abstract:}
A unary constraint (on the Boolean domain) is a function from $\{0,1\}$ to the set of real numbers. A free use of
auxiliary unary constraints given besides input instances has proven to be useful in establishing a complete classification of the computational complexity of approximately solving weighted counting Boolean constraint satisfaction problems (or \#CSPs).  In particular, two special constant  unary constraints are a key to an arity reduction of arbitrary constraints, sufficient for the desired classification.
In an exact counting model, both constant unary constraints are always assumed to be available since they can be eliminated efficiently using
an arbitrary nonempty set of constraints. In contrast,
we demonstrate in an approximate counting model, that at least one of them is efficiently approximated and thus eliminated approximately
by a nonempty constraint set.
This fact directly leads to an efficient construction of polynomial-time randomized approximation-preserving Turing reductions (or AP-reductions) from \#CSPs with designated constraints
to any given \#CSPs composed of symmetric real-valued  constraints of arbitrary arities even in the presence of arbitrary extra unary constraints.

\s

\n{\bf Keywords:} counting constraint satisfaction problem,  AP-reducible,
effectively T-constructible, constant unary constraint, symmetric constraint, algebraic real number, p-convergence
\end{quote}

%%%%%%%%%%%%%%%%%
%%%%%%%%%%%%%%%%%
\section{Roles of Constant Unary Constraints}\label{sec:introduction}

{\em Constraint satisfaction problems} (or {\em CSPs}, in short) are combinatorial problems that have been ubiquitously found in real-life situations.
The importance of these problems have led recent intensive studies from various aspects: for instance, decision CSPs \cite{DF03,Sch78}, optimization CSPs \cite{Cre95,Yam11b}, and counting CSPs \cite{CL07,CH96,DGJ09,Yam10a}.
Driven by theoretical and practical interests, in this paper, we are particularly focused on {\em counting Boolean CSPs} (abbreviated as \#CSPs) whose goal is to count the number of variable assignments that satisfy all given Boolean-valued constraints defined over a fixed series of Boolean variables.
The problem of counting the number of Boolean assignments that satisfy each given propositional formula, known as \#SAT (counting satisfiability problem), is a typical counting CSP with three Boolean-valued constraints, $AND$, $OR$, and $NOT$. As this example demonstrates,
in most real-life applications, all available constraints are
pre-determined. Hence, we naturally fix a collection of ``allowed'' constraints, say, $\FF$  and
wish to solve every \#CSP whose constraints are all chosen from $\FF$.
Such a counting problem is conventionally denoted $\sharpcsp(\FF)$ and this notation will be used throughout this paper.
Creignou and Hermann \cite{CH96} first examined the computational complexity of {\em exactly} counting solutions of unweighted \#CSPs.  Recently,
Dyer, Goldberg, and Jerrum \cite{DGJ10} studied the computational complexity of {\em approximately} computing the number of solutions of unweighted \#CSPs
using a technical reduction, known as {\em polynomial-time randomized approximation-preserving Turing reduction} (or {\em AP-reduction}, hereafter), whose formulation is originated from  \cite{DGGJ04} and it will be explained in details through Section \ref{sec:randomized-scheme}.

In a more interesting case of {\em weighted \#CSPs}, the values of constraints are expanded from Boolean values to more general values, and each  weighted \#CSP asks for the sum, over all possible assignments for Boolean variables, of products of the output values of all given constraints.
Earlier, Cai, Lu, and Xia \cite{CLX09x} gave a complete classification of
complex-weighted \#CSPs restricted on a given set $\FF$ of constraints
according to the computational complexity of {\em exactly} solving them.
Another complete classification regarding the complexity of {\em approximately} solving complex-weighted \#CSPs was presented by Yamakami \cite{Yam10a} when allowing a free use of auxiliary unary (\ie arity-1)  constraints besides initially given input constraints.
More precisely, let $\UU$ denote the set of all unary constraints. Given an arbitrary constraint $f$, the free use of auxiliary unary constraints makes  $\#\mathrm{SAT}_{\complex}$ (a complex extension of $\#\mathrm{SAT}$) AP-reducible to $\sharpcsp(f,\UU)$ unless $f$ is factored into three categories of constraints: the binary equality, the binary disequality, and unary constraints \cite{Yam10a}.
All constraints factored into constraints of those categories form a special set $\ED$. The aforementioned fact establishes the following complete classification of the approximation complexity of weighted \#CSPs in the presence of $\UU$.

\begin{theorem}{\em \cite[Theorem 1.1]{Yam10a}}\label{dichotomy-theorem}
Let $\FF$ be any set of complex-valued constraints. If $\FF\subseteq \ED$, then  $\sharpcsp(\FF,\UU)$ is solvable in polynomial time; otherwise, it is AP-reduced from $\#\mathrm{SAT}_{\complex}$.
\end{theorem}

In this particular classification, the free use of auxiliary unary constraints provide  enormous power that makes it possible to establish a ``dichotomy'' theorem  beyond a ``trichotomy'' theorem of Dyer \etalc~\cite{DGJ10} for Boolean-valued constraints (or simply, {\em Boolean constraints}).
The proof of Theorem \ref{dichotomy-theorem}  in \cite{Yam10a} employed  two technical notions:  ``factorization'' and ``T-constructibility.''
Limited to unweighted \#CSPs, on the contrary, a key to the proof of the trichotomy theorem of \cite{DGJ10} is an efficient approximation of so-called {\em constant unary constraints}, conventionally denoted\footnote{A bracket notation $[x,y]$ denotes a unary function $g$ satisfying $g(0)=x$ and $g(1)=y$. Similarly, $[x,y,z]$ expresses a binary function $g$ for which  $g(0,0)=x$, $g(0,1)=g(1,0)=y$, and $g(1,1)=z$.}  $\Delta_0=[1,0]$ and $\Delta_1=[0,1]$.
A significant use of the constant unary constraints is a technique known as
{\em pinning}, with which we can make an arbitrary variable pinned down to a particular value, reducing the associated constraints of high arity
to those of lower arity.
To see this arity reduction, let us consider, for example, an arbitrary constraint $f$ of the form $[x,y,z]$ with three Boolean variables $x_1,x_2,x_3$.
When we pin a variable $x_1$ down to $0$ (resp., $1$) in $f(x_1,x_2,x_3)$,
we immediately obtain another constraint of the form $[x,y]$ (resp., $[y,z]$).
Therefore, an efficient approximation of those special constraints helps us first analyze the approximation complexity of $\sharpcsp(\FF,\Delta_{i_0})$ for an appropriate index $i_0\in\{0,1\}$ by the way of pinning and then eliminate $\Delta_{i_0}$ completely to obtain a desired classification theorem for $\sharpcsp(\FF)$. Their proof of approximately eliminating the constant unary constraints  is based on basic properties of Boolean arithmetic and it is not entirely clear that we can expand their proof to a non-Boolean case. Therefore, it is natural for us to raise a question of whether we can obtain a similar elimination theorem for $\sharpcsp(\FF)$ even when $\FF$ is composed of real-valued constraints.
In the following theorem, we wish to claim that at least one of the constant unary constraints  is always eliminated approximately.  This claim can be sharply  contrasted with the case of exact counting of $\sharpcsp(\FF)$,
in  which  $\Delta_0$ and $\Delta_1$ are {\em both} eliminated deterministically by a technique known as {\em polynomial interpolation}.

\begin{theorem}\label{key-Delta-elimination}
For any nonempty set $\FF$ of real-valued constraints, there exists a constant unary constraint $h\in\{\Delta_0,\Delta_1\}$ for which $\sharpcsp(h,\FF)$ is AP-equivalent to $\sharpcsp(\FF)$ (namely, $\sharpcsp(h,\FF)$ is AP-reducible to $\sharpcsp(\FF)$ and vice versa).
\end{theorem}

Under a certain set of explicit conditions (given in Proposition \ref{Delta-removal-arity-all}), we further prove that $\Delta_0$ and $\Delta_1$ are simultaneously eliminated even in an approximation sense.

When the values of constraints in $\FF$ are all limited to Boolean values, Theorem \ref{key-Delta-elimination} is exactly \cite[Lemma 16]{DGJ10}. For real-valued constraints, however, we need to develop a quite different argument from \cite{DGJ10} to prove this theorem. An important ingredient of our proof, described in Section \ref{sec:const-unary-const}, is an efficient estimation of a lower bound of an arbitrary multi-variate polynomial in the values of given constraints. However, since our constraints can output negative real values, the polynomial may possibly produce arbitrary small values, and thus we cannot find a polynomial-time computable lower bound. To avoid encountering such an unwanted situation, we dare to restrict our attention onto {\em algebraic real numbers}. In the rest of this paper, all real numbers will be limited to algebraic numbers.

As a natural application of Theorem \ref{key-Delta-elimination}, we give an alternative proof to our classification theorem (Theorem \ref{dichotomy-theorem}) for {\em symmetric
real-weighted \#CSPs} when arbitrary unary constraints are freely available.
Using the constant unary constraints, we can conduct the aforementioned
arity reductions.
Since Theorem \ref{key-Delta-elimination} guarantees the availability of only one of  $\Delta_0$ and $\Delta_1$, we need to demonstrate such arity reductions of target constraints even when $\Delta_0$ and $\Delta_1$ are separately given for free. Furthermore, we intend to build such reductions with no use of  auxiliary unary constraint.

Our alternative proof proceeds roughly as follows. In the first step, we
recognize constraints $g$ of the following three  special forms: $[0,y,z]$ and $[x,y,0]$ with $x,y,z>0$ and  $[x,y,z]$ with $x,y,z>0$ as well as $xz\neq y^2$. The constraints $g$ of those forms become crucial elements of
our later analyses
because, when auxiliary unary constraints are available for free, $\sharpcsp(g,\UU)$ is  computationally at least as hard as $\#\mathrm{SAT}$ with respect to the AP-reducibility (Lemma \ref{OR-and-B-Yam10a}).

In the second step, we isolate a set $\FF$ of constraints  whose corresponding counting problem $\sharpcsp(\FF,\GG)$ is AP-reduced from a specific
problem $\sharpcsp(g,\GG)$ for an arbitrary set $\GG$ of constraints with no use of extra unary constraints. To be more exact, we wish to establish the following specific AP-reduction from $\sharpcsp(g,\GG)$ to $\sharpcsp(\FF,\GG)$.

\begin{theorem}\label{main-theorem}
Let $\FF$ be any set of symmetric real-valued constraints of arity at least $2$. If either $\FF\subseteq \DG\cup\ED_{1}^{(+)}$ or $\FF\subseteq\DG^{(-)}\cup\ED_{1}\cup\AZ \cup\AZ_1\cup\BB_0$ holds, then $\sharpcsp(\FF)$ is polynomial-time solvable. Otherwise, $\sharpcsp(\FF,\GG)$ are AP-reduced from $\sharpcsp(g,\GG)$ for any constraint set $\GG$, where $g$ is an appropriate constraint of one of the three
special forms described above.
\end{theorem}

In Theorem \ref{main-theorem}, the constraint set $\DG$ consists of {\em degenerate} constraints, $\ED_{1}$ indicates a set of equality and disequality, $\AZ$ contains specific symmetric constraints having alternating zeros, $\AZ_1$ is similar to $\AZ$ but requiring  alternating zeros, ``plus'' signs, and ``minus'' signs, and $\BB_0$ is composed of special constraints of non-zero entries. Two additional sets $\DG^{(-)}$ and $\ED_{1}^{(+)}$ are naturally induced from $\DG$ and $\ED_{1}$, respectively.  For their precise definitions,  refer to Section \ref{sec:computability}.

\sloppy
In the third step, we recognize distinctive behaviors of two constraint sets $\DG\cup\ED_1^{(+)}$ and $\DG^{(-)}\cup \ED_1\cup\AZ \cup\AZ_1\cup \BB_0$.
The counting problems $\sharpcsp(\DG,\ED_{1}^{(+)})$ and $\sharpcsp(\DG^{(-)},\ED_{1},\AZ,\AZ_1,\BB_0)$ are both solvable in polynomial time \cite{CLX09x,GGJ+10}. In the presence of the auxiliary set $\UU$ of arbitrary unary constraints, the problem $\sharpcsp(\DG,\ED_{1}^{(+)},\UU)$,  which essentially equals $\sharpcsp(\ED,\UU)$,
remains solvable in polynomial time;
on the contrary, as a consequence of Theorem \ref{main-theorem}, the problem $\sharpcsp(\DG^{(-)},\ED_{1},\AZ,\AZ_1,\BB_0,\UU)$ is AP-reduced from $\sharpcsp(g,\UU)$ for an appropriately chosen $g$ of the prescribed form.

In the final step, since $\#\mathrm{SAT}$ is AP-reducible to $\sharpcsp(g,\UU)$ \cite{Yam10a}, the above results immediately imply Theorem \ref{dichotomy-theorem} for symmetric real-weighted \#CSPs. The above  argument exemplifies that the free use of auxiliary unary constraints can be  made only in the third step. The detailed argument is found in Section \ref{sec:computability}.

A heart of our proof is an efficient, approximate transformation  (called {\em effective T-constructibility}) of a target constraint from a given set of constraints. This effective T-constructibility is a powerful tool in showing AP-reductions between two counting problems. Since our constructibility can
locally modify underlying structures of input instances, this simple tool makes it possible to introduce an auxiliary constraint set $\GG$ in Theorem \ref{main-theorem}. A prototype of  this technical tool first appeared in \cite{Yam10a} and was further extended or modified in \cite{Yam10b,Yam11a}.

%%%%
\paragraph{Comparison of Proof Techniques:}
Dyer \etalc~\cite{DGJ10} used a notion of ``simulatability'' to demonstrate  the approximate elimination of the constant unary constraints using any given set of Boolean constraints.  Our proof of Theorem \ref{key-Delta-elimination}, however, employs a notion of effectively T-constructibility.
While a key proof technique used in \cite{Yam10a} to  prove Theorem \ref{dichotomy-theorem}  is the factorization of constraints,  our proof of Theorem  \ref{main-theorem} (which leads to Theorem \ref{dichotomy-theorem})
in Section \ref{sec:computability} makes a heavy use of the constant unary constraints.  Furthermore, our proof is quite elementary because it proceeds by examining {\em all} possible forms of a target constraint. This fact makes the proof cleaner and more straightforward to follow.

%%%%%
%%%%%
\section{Fundamental Notions and Notations}\label{sec:preliminaries}

We will explain basic concepts that are necessary to read through the rest of this paper. First, let $\nat$ denote the set of all {\em natural numbers} (\ie nonnegative integers) and let $\real$ be  the set of all {\em real numbers}. For convenience, define $\nat^{+} = \nat-\{0\}$  and,  for each number $n\in\nat^{+}$,  $[n]$ stands for the {\em integer interval} $\{1,2,\ldots,n\}$.

Because our results heavily rely on Lemma \ref{constructibility}(3), we need to limit our attention within {\em algebraic real numbers}. For this purpose, a  special notation $\algebraic$ is used to indicate the set of all algebraic real numbers. To simplify our terminology throughout the paper, whenever we refer to ``real numbers,'' we actually mean ``algebraic real numbers.''

%%%%%
%%%%%
\subsection{Constraints and \#CSPs}

The term ``constraint of arity $k$'' always  refers to a function mapping
the set $\{0,1\}^{k}$ of binary strings of length $k$ to $\algebraic$.
Assuming the standard lexicographic ordering on the set $\{0,1\}^{k}$,
we conveniently express $f$ as a {\em row-vector} consisting of its output values; for instance, when $f$ has arity $2$, it is expressed as $(f(00),f(01),f(10),f(11))$.
Given any $k$-ary constraint $f=(f_1,f_2,\ldots,f_{2^k})$ in a vector form,  the notation $\|f\|_{\infty}$ means    $\max_{i\in[2^k]}\{|f_i|\}$.
A $k$-ary constraint $f$ is called {\em symmetric} if, for every input $x$ in $\{0,1\}^k$, the value $f(x)$ depends only on the Hamming weight (\ie the number of $1$'s) of the input $x$; otherwise, $f$ is called {\em asymmetric}.
For any symmetric constraint $f$ of arity $k$, we also use a succinct notation  $[f_0,f_1,\ldots,f_k]$ to express $f$, where each entry $f_i$ expresses the value of $f$ on inputs of
Hamming weight $i$. For instance, if $f=[f_0,f_1,f_2]$ is of arity two, then it holds that $f_0=f(00)$, $f_1=f(01)=f(10)$, and $f_2=f(11)$.
Of all symmetric constraints, we recognize two special unary constraints,  $\Delta_0=[1,0]$ and $\Delta_1=[0,1]$, which are called {\em constant unary constraints}.

Restricted to a set $\FF$ of constraints, a
{\em real-weighted (Boolean) \#CSP}, conventionally
denoted $\sharpcsp(\FF)$, takes a {\em finite} set $\Omega$ composed of elements of the form $\pair{h,(x_{i_1},x_{i_2},\ldots,x_{i_k})}$, where $h\in\FF$ is a function  on $k$ Boolean variables $x_{i_1},x_{i_2},\ldots,x_{i_k}$ in $X=\{x_1,x_2,\ldots,x_n\}$ with  $i_1,\ldots,i_k\in[n]$, and its goal is to compute the real value
\begin{equation}\label{eqn:csp-def}
\csp_{\Omega} =_{def} \sum_{x_1,x_2,\ldots,x_n\in\{0,1\}}
\prod_{\pair{h,x}\in \Omega} h(x_{i_1},x_{i_2},\ldots,x_{i_k}),
\end{equation}
where $x$ denotes $(x_{i_1},x_{i_2},\ldots,x_{i_k})$. To illustrate $\Omega$  graphically, we view it as a {\em labeled undirected bipartite graph} $G=(V_1|V_2,E)$ whose nodes in $V_1$ are labeled
distinctively by $x_1,x_2,\ldots,x_n$ in $X$ and nodes in $V_2$ are labeled by constraints $h$ in $\FF$ such that, for each pair  $\pair{h,(x_{i_1},x_{i_2},\ldots,x_{i_k})}$,  there are $k$ edges between
an associated node labeled $h$ and the nodes labeled $x_{i_1},x_{i_2},\ldots,x_{i_k}$. The labels of nodes are formally specified by a {\em labeling function} $\pi:V_1\cup V_2\rightarrow X\cup\FF$ with $\pi(V_1)\subseteq X$ and $\pi(V_2)\subseteq \FF$ but we often omit it from the description of $G$ for simplicity.  When $\Omega$ is viewed as this
special bipartite graph, it is called a {\em constraint frame}   \cite{Yam10a,Yam10b}.  More formally, a constraint frame $\Omega=(G,X|\FF',\pi)$ is composed of an undirected  bipartite graph $G$ with its associated labeling function $\pi:V_1\cup V_2\rightarrow X\cup \FF'$, a variable set $X=\{x_1,x_2,\ldots,x_n\}$, and a {\em finite} set $\FF'\subseteq \FF$.

To simplify later descriptions, we wish to use the following simple rule of abbreviation.
For instance, when $f$ is a constraint and both $\FF$ and $\GG$ are constraint sets, we write  $\sharpcsp(f,\FF,\GG)$ to mean $\sharpcsp(\{f\}\cup \FF\cup\GG)$.

In the subsequent sections, we will use the following succinct notations.
Let $f$ be any constraint of arity $k\in\nat^{+}$.
Given any index $i\in[k]$ and any bit $c\in\{0,1\}$,
the notation  $f^{x_i=c}$ stands for the function $g$ of arity $k-1$ satisfying that $g(x_1,\ldots,x_{i-1},x_{i+1},\ldots,x_k) = f(x_1,\ldots,x_{i-1},c,x_{i+1},\ldots,x_k)$ for every $(k-1)$-tuple $(x_1,\ldots,x_{i-1},x_{i+1},\ldots,x_k)\in\{0,1\}^{k-1}$.
For any two distinct
indices $i,j\in[k]$, we denote by $f^{x_i=x_j}$ the function $g$ defined as  $g(x_1,\ldots,x_{i-1},x_{i+1},\ldots,x_k) =
 f(x_1,\ldots,x_{i-1},x_{j},x_{i+1},\ldots,x_k)$ for every $k$-tuple $(x_1,x_2,\ldots,x_k)\in\{0,1\}^k$.
Finally, let $f^{x_i=*}$ express the function $g$ defined by   $g(x_1,\ldots,x_{i-1},x_{i+1},\ldots,x_k) = \sum_{c\in\{0,1\}}f(x_1,\ldots,x_{i-1},c,x_{i+1},\ldots,x_k)$ for every $(k-1)$-tuple $(x_1,\ldots,x_{i-1},x_{i+1},\ldots,x_k)\in\{0,1\}^{k-1}$.

%%%%%%%%%%%%%%%%%
\subsection{FP$_{\algebraic}$ and AP-Reducibility}\label{sec:randomized-scheme}

To connect our results (particularly, Theorems \ref{key-Delta-elimination}
and \ref{main-theorem}) to Theorem \ref{dichotomy-theorem}, we follow notational conventions used in \cite{Yam10a,Yam10b}.
First, $\fp_{\algebraic}$ denotes the collection of all {\em $\algebraic$-valued}
functions that can be computed deterministically in polynomial time.

Let $F$ be any function mapping $\{0,1\}^*$ to $\algebraic$ and let $\Sigma$ be any nonempty finite alphabet. A {\em randomized approximation scheme} (or RAS, in short) for $F$ is a randomized algorithm that takes a standard input $x\in\Sigma^*$ together with an error tolerance parameter $\varepsilon\in(0,1)$, and outputs values $w$ with probability at least $3/4$ for which
\begin{equation}\label{eqn:approx}
\min\{2^{-\varepsilon}F(x), 2^{\varepsilon}F(x)\} \leq w \leq \max\{2^{-\varepsilon}F(x), 2^{\varepsilon}F(x)\}.
\end{equation}

Given two arbitrary real-valued functions $F$ and $G$, a {\em polynomial-time randomized  approximation-preserving Turing reduction}
(or {\em AP-reduction}) from $F$ to $G$ \cite{DGGJ04}
is a randomized algorithm $M$ that takes a pair $(x,\varepsilon)\in\Sigma^*\times(0,1)$ as input,
accesses an oracle,
and satisfies the following three conditions:
(i) when the oracle is an arbitrary RAS $N$ for $G$,
$M$ is always an RAS for $F$;
(ii) every oracle call made by $M$ is of the form $(w,\delta)\in\Sigma^*\times(0,1)$ with $1/\delta \leq poly(|x|,1/\varepsilon)$ and its answer is the outcome of $N$ on $(w,\delta)$; and (iii) the running time of $M$ is upper-bounded
by a certain polynomial in $(|x|,1/\varepsilon)$, which is not dependent of  the choice of $N$. If such an AP-reduction exists, then we also say that $F$ is {\em AP-reducible} to $G$ and we write $F\APreduces G$.
If both $F\APreduces G$ and $G\APreduces F$ hold, then $F$ and $G$ are
said to be {\em AP-equivalent} and we use the special notation $F\APequiv G$.

\begin{lemma}\label{AP-property}
For any functions $F_1,F_2,F_3:\{0,1\}^*\rightarrow\algebraic$, the following properties hold.
\begin{enumerate}\vs{-1}
\item $F_1\APreduces F_1$.
\vs{-2}
\item If $F_1\APreduces F_2$ and $F_2\APreduces F_3$, then $F_1\APreduces F_3$.
\end{enumerate}
\end{lemma}

%%%%%%%%%%%%%%%%%
%%%%%%%%%%%%%%%%%
\subsection{Effective T-Constructibility}\label{sec:constructibility}

Our goal in the subsequent sections  is to prove our main theorems, Theorems \ref{key-Delta-elimination} and \ref{main-theorem}. For their desired proofs,  we will introduce a fundamental notion of {\em effective T-constructibility}, whose underlying idea comes from a graph-theoretical formulation of {\em limited T-constructibility} \cite{Yam10b}.

Let us start with the definitions of ``representation'' and ``realization''
in   \cite{Yam10b}.
Let $f$ be any constraint of arity $k$. We say that an undirected bipartite graph $G=(V_1|V_2,E)$ (together with a labeling function $\pi$) {\em represents} $f$ if $V_1$ consists only of $k$ nodes labeled with $x_1,\ldots,x_k$, which may possibly have a certain number of
dangling edges,\footnote{A dangling edge is obtained from an edge by deleting exactly one end of this edge. These dangling edges are treated as ``normal'' edges, and therefore  the degree of each node must count dangling
edges as well.}
and $V_2$ contains only a node labeled
$f$ to whom each node $x_i$ is adjacent.
Given  a set $\GG$ of constraints, a graph $G=(V_1|V_2,E)$ is said to
{\em  realize $f$ by $\GG$} if the following four
conditions are met simultaneously:
\begin{enumerate}
\item[(i)] $\pi(V_2)\subseteq \GG$,
\vs{-2}
\item[(ii)] $G$ contains at least $k$ nodes having the labels  $x_1,\ldots,x_k$, possibly together with nodes associated with other variables, say, $y_1,\ldots,y_m$; namely,
    $V_1=\{x_1,\ldots,x_k,y_1,\ldots,y_m\}$,
\vs{-2}
\item[(iii)]  only the nodes $x_1,\ldots,x_k$ are allowed to have dangling edges, and
\vs{-2}
\item[(iv)] $f(x_1,\ldots,x_k) = \lambda \sum_{y_1,\ldots,y_m\in\{0,1\}} \prod_{w \in V_2} f_{w}(z_1,\ldots,z_d)$ for an appropriate constant  $\lambda\in\algebraic-\{0\}$, where $f_w$ denotes a constraint $\pi(w)$ and $z_1,\ldots,z_d\in  V_1$.
\end{enumerate}

The {\em sign function}, denoted $sgn$, is defined as follows. For any real number $\lambda$, we set $sgn(\lambda)=+1$ if $\lambda>0$, $sgn(\lambda)=0$ if $\lambda=0$, and $sgn(\lambda)=-1$ if $\lambda<0$.
An infinite series $\Lambda = (g_1,g_2,g_3,\ldots)$ of arity-$k$ constraints
is called a {\em p-convergence series}\footnote{At a quick glance, the approximation scheme of Eq.(\ref{eqn:convergence}) appears quite differently from that of Eq.(\ref{eqn:approx}). However, by setting $\varepsilon = \lambda^m$, the value $1+\varepsilon$ approximately equals $2^{\varepsilon}$ and $1-\varepsilon$ is also close to $2^{-\varepsilon}$ for any sufficiently large number $m$.}
for a target constraint $f=(r_1,r_2,\ldots,r_{2^k})$ of arity $k$ if there exist a constant $\lambda\in(0,1)$ and a deterministic Turing machine (abbreviated as DTM) $M$ running in polynomial time such that, for every number $m\in\nat^{+}$, (i) $M$ takes an input of the form $1^m$ and outputs a  complete description of the constraint $g_m$ in a row-vector form  $(z_1,z_2,\ldots,z_{2^k})$, (ii) for every $k$-tuple $x\in\{0,1\}^k$, if $f(x)\neq0$, then $sgn(f(x))=sgn(g_m(x))$, and (iii) for every index $i\in[2^k]$, if $r_i\neq0$, then
\begin{equation}\label{eqn:convergence}
\min\{(1+\lambda^m)z_i,(1-\lambda^m)z_i\}\leq r_i\leq \max\{(1+\lambda^m)z_i,(1-\lambda^m)z_i\},
\end{equation}
and otherwise, $|z_i|\leq \lambda^m$.

We then define the effective T-constructibility of a given finite set of constraints.

\begin{definition}\label{def:constructibility}
(effective T-constructibility)
Let $\FF$ and $\GG$ be any two finite sets of constraints. We say that $\FF$ is {\em effectively T-constructible} from $\GG$ if there exists a finite series  $(\FF_1,\FF_2,\ldots,\FF_n)$ of finite constraint sets (which is
succinctly called a {\em generating series} of $\FF$ from $\GG$) such that
\begin{enumerate}\vs{-1}
\item[(i)] $\FF = \FF_1$ and $\GG=\FF_n$, and
\vs{-2}
\item[(ii)] for each adjacent pair $(\FF_i,\FF_{i+1})$, where $i\in[n-1]$, one of Clauses (I)--(II) should hold.
\end{enumerate}\vs{-1}

(I) For every constraint $f$ of arity $k$ in $\FF_i$ and for any finite graph $G$ representing $f$ with distinct variables $x_1,\ldots,x_k$, there exists another finite graph $G'$ satisfying the following two conditions:
\begin{enumerate}\vs{-1}
\item[(i')] $G'$ realizes $f$ by $\FF_{i+1}$, and
\vs{-2}
\item[(ii')] $G'$ maintains the same dangling edges as $G$ does.
\end{enumerate}\vs{-1}

(II) Let $\FF_{i+1}=\{g_1,g_2,\ldots,g_d\}$. For every constraint $f$ of arity $k$ in $\FF_i$, there exist a
p-convergence series $\Lambda=(f_1,f_2,\ldots)$ of arity-$k$ constraints and a polynomial-time DTM $M$ such that, for every number $m\in\nat^{+}$, (a) $M$ takes an input of the form $(1^m,G,(g_1,g_2,\ldots,g_d))$, where $G$ represents $f_m$ with distinct variables $x_1,\ldots,x_k$ and each $g_j$ is described in a row-vector form, and (b) $M$ outputs a bipartite graph $G_m$ such that
\begin{enumerate}\vs{-1}
\item[(i'')] $G_m$ realizes $f_m$ by $\FF_{i+1}$ and
\vs{-2}
\item[(ii'')] $G_m$ maintains the same dangling edges as $G$ does.
\end{enumerate}\vs{-1}
\end{definition}

When $\FF$ is effectively T-constructible from $\GG$, we write $\FF\econst \GG$. We are particularly interested in the case where $\FF$ is a singleton $\{f\}$, and we succinctly write $f\econst\GG$. Moreover, when $\GG$ is also a singleton $\{g\}$, we further write $f\econst g$.

\begin{lemma}\label{constructibility}
Let $\FF_1$, $\FF_2$, and $\FF_3$ be three finite constraint sets. Let $\GG$ be an arbitrary set of constraints.
\begin{enumerate}\vs{-1}
\item $\FF_1\econst \FF_1$.
\vs{-2}
\item $\FF_1\econst \FF_2$ and $\FF_2\econst \FF_3$ imply $\FF_1\econst \FF_3$.
\vs{-2}
\item If $\FF_1\econst \FF_2$, then $\sharpcsp(\FF_1,\GG)\APreduces \sharpcsp(\FF_2,\GG)$.
\end{enumerate}
\end{lemma}

It is important to note that, since our constraints are permitted to output negative values, the use of {\em algebraic real numbers} for the constraints may be necessary in the proof of Lemma \ref{constructibility} because the proof heavily relies on an explicit lower bound estimation of arbitrary polynomials over algebraic numbers.

In later arguments, a use of effective T-constructibility will play an essential role because a relation $f\econst g$ leads to $\sharpcsp(f,\GG)\APreduces \sharpcsp(g,\GG)$ for any constraint set $\GG$  by Lemma \ref{constructibility}(3), whereas a relation $\sharpcsp(f)\APreduces \sharpcsp(g)$ in general does not imply $\sharpcsp(f,\GG)\APreduces \sharpcsp(g,\GG)$.

For the readability, we postpone the proof of Lemma \ref{constructibility} until Appendix.

%%%%%%%%%%%%%%%%%%%%%%
%%%%%%%%%%%%%%%%%%%%%%
\section{Approximation of the Constant Unary Constraints}\label{sec:const-unary-const}

Let us prove our first main theorem---Theorem \ref{key-Delta-elimination}---which states that, given an arbitrary set of constraints,  we can efficiently approximate at least one of the two constant unary constraints. The theorem allows us to utilize such a constraint freely for a further analysis of constraints in Section \ref{sec:computability}.

%%%%%%%%%%%%%%%%%%%%%%%%
%%%%%%%%%%%%%%%%%%%%%%%%
%%%%%%%%%%%%%%%%
\subsection{Notion of Complement Stability}\label{sec:complement-stable}

To prove Theorem \ref{key-Delta-elimination}, we will first introduce two useful notions regarding a certain ``symmetric'' nature of a given constraint.
A $k$-ary constraint $f$  is said to be {\em complement invariant} if $f(x_1,\ldots,x_k) = f(x_1\oplus1,\ldots,x_k\oplus1)$ holds for every input tuple  $(x_1,\ldots,x_k)$ in $\{0,1\}^k$, where the notation $\oplus$ means the {\em (bitwise) XOR}.
In contrast, we say that $f$ is {\em complement anti-invariant} if, for every input $(x_1,\ldots,x_k)\in\{0,1\}^k$, $f(x_1,\ldots,x_k) = - f(x_1\oplus1,\ldots,x_k\oplus1)$ holds.  For instance, $f=[1,1]$ is complement  invariant and $f'=[1,0,-1]$ is complement anti-invariant. In addition, we say that $f$ is {\em complement stable} if $f$ is either complement invariant or complement anti-invariant. A constraint set $\FF$  is {\em complement stable} if every constraint in $\FF$ is complement stable.
In the case where $f$ (resp., $\FF$) is not complement stable, by contrast,
we conveniently call it {\em complement unstable}.

We will split Theorem \ref{key-Delta-elimination} into two separate statements, as shown in Proposition \ref{Delta-removal-arity-all}, depending on whether or not a given nonempty set $\FF$ of constraints is complement stable.

\begin{proposition}\label{Delta-removal-arity-all}
Let $\FF$ be any nonempty set of constraints and let $f$ be any constraint of arity $k$ with $k\geq1$.
\renewcommand{\labelitemi}{$\circ$}
\begin{enumerate}
  \setlength{\topsep}{-2mm}%
  \setlength{\itemsep}{1mm}%
  \setlength{\parskip}{0cm}%

\item If $\FF$ is complement stable, then $\sharpcsp(\Delta_i,\FF)\APequiv \sharpcsp(\FF)$ holds for every index $i\in\{0,1\}$.

\item Assume that $f$ is complement unstable. If  $f$ satisfies one of two conditions ($a$)--($b$) given below, then $\Delta_i\econst f$ holds for all indices $i\in\{0,1\}$. Otherwise, there exists at least one index $i\in\{0,1\}$ for which $\Delta_i\econst  f$ holds.

\begin{enumerate}\vs{-1}
  \setlength{\topsep}{-1mm}%
  \setlength{\itemsep}{1mm}%
  \setlength{\parskip}{0cm}%

\item $k\geq2$ and $|z_1| = |z_{2^k}|$.

\item $k\geq2$  and
either $(|z_1|-|z_{2^k}|)(|z_1+z_j| - |z_{2^k-j+1}+z_{2^k}|)<0$ or
$(|z_1|-|z_{2^k}|)(|z_1+z_{2^k-j+1}| - |z_j+z_{2^k}|)<0$ holds  for a certain index $j\in[2^{k-1}]-\{1\}$.
\end{enumerate}
\end{enumerate}
\end{proposition}

%%%%%%%
%%%%%%%

Notice that Proposition \ref{Delta-removal-arity-all} together with Lemma \ref{constructibility}(3) implies Theorem \ref{key-Delta-elimination}. Proposition \ref{Delta-removal-arity-all}(1) can be proven rather easily, as presented below, whereas Proposition \ref{Delta-removal-arity-all}(2) requires a slightly more complicated argument.

\begin{proofof}{Proposition \ref{Delta-removal-arity-all}(1)}
In the following proof, we will deal only with $\Delta_0$, because the other case is similarly handled. Let $\FF$ be any nonempty set of constraints and take any input instance $\Omega$, in the form of constraint frame $(G,X|\FF',\pi)$ with $\FF'\subseteq \FF$,  given to the counting problem $\sharpcsp(\Delta_0,\FF)$. If $\Delta_0\in\FF$, then Proposition \ref{Delta-removal-arity-all}(1) is trivially true.
Henceforth, we assume that $\Delta_0\not\in\FF$.
Let us consider the case where $\FF$ is complement anti-invariant. Since the other case where $\FF$ is complement invariant is essentially the same, we omit the case.

To simplify our proof, we modify $\Omega$ as follows. First, we merge all variable nodes (\ie nodes with ``variable'' labels) adjacent to nodes labeled $\Delta_0$ into a single node having a fresh variable label. If there are more than one adjacent nodes with the label $\Delta_0$, then we delete all those nodes except for one node. After this modification, we always assume that there is exactly one node, say, $v_0$ whose label is $\Delta_0$. Now, let $v_1$ be a unique node
adjacent to $v_0$ and let $x_0$ be its variable label. For simplicity,
we keep the same notation $\Omega$ for the constraint frame obtained by this modification.

Let $m$ denote the total number of {\em nodes} in $\Omega$ whose labels are constraints  in $\FF$. By simply removing the node $v_0$ having the label $\Delta_0$ from $\Omega$, we obtain  another instance, say,  $\Omega'$, which  is obviously an input instance to $\sharpcsp(\FF)$.
Using basic properties of complement anti-invariance, we wish to prove
the following equality:
\begin{equation}\label{csp-formula}
\csp_{\Omega'} = \csp_{\Omega} + (-1)^m\csp_{\Omega}.
\end{equation}
Let us consider any ``partial'' assignment $\sigma$  to all variables appearing in $\Omega'$ except for $x_0$, that is, $\sigma:X-\{x_0\}\rightarrow\{0,1\}$.
Associated with $\sigma$,  we introduce
two corresponding Boolean assignments $\sigma_0$ and $\sigma_1$. Firstly, we obtain $\sigma_0$ from $\sigma$ by additionally assigning $0$ to $x_0$. Now, we assume that $\sigma_0$ is an satisfying assignment for $\Omega'$.  Secondly, let $\sigma_1$ be defined by assigning  $1$ to $x_0$ and $1-\sigma(z)$ to all the other variables $z$.
Note that $csp_{\Omega}$ is calculated over all assignments $\sigma_0$ induced from any partial assignments $\sigma$. Similarly, to compute $csp_{\Omega'}$, is is enough to consider all assignments $\sigma_0$ and $\sigma_1$.
Since all constraints in $\Omega'$ are complement anti-invariant,
the product of the values of all constraints by $\sigma_1$ equals
$(-1)^m$ times the product of all constraints' values by $\sigma_0$.
This establishes Eq.(\ref{csp-formula}).

If $m$ is even, then we immediately obtain the equation $\csp_{\Omega} = \frac{1}{2}\csp_{\Omega'}$ from Eq.(\ref{csp-formula}).
Next, assume that $m$ is odd and choose any constraint $g$ that is complement anti-invariant in $\FF$. We further modify $\Omega'$ into $\Omega''$ as follows. Letting $g$ be of arity $k$, we prepare a new variable, say, $x$ and add to $\Omega'$ a new element $\pair{g,(x,x,\ldots,x)}$, which essentially behaves
as $e\cdot [1,-1]$ for a certain constant $e\neq0$.  A similar argument  for  Eq.(\ref{csp-formula}) can prove that
\begin{equation*}
\csp_{\Omega''} = e\cdot \csp_{\Omega} + (-1)^{m+1}e\cdot\csp_{\Omega} = 2e\cdot \csp_{\Omega}.
\end{equation*}
Thus, from the value  $\csp_{\Omega''}$, we can efficiently compute $\csp_{\Omega}$, which equals  $\frac{1}{2e}\csp_{\Omega''}$.
The two equations $\csp_{\Omega}=\frac{1}{2}csp_{\Omega'}$ and $\csp_{\Omega}=\frac{1}{2e}csp_{\Omega''}$ clearly
establish an AP-reduction from  $\sharpcsp(\Delta_i,\FF)$ to $\sharpcsp(\FF)$.

Since the other direction, $\sharpcsp(\FF)\APreduces \sharpcsp(\Delta_i,\FF)$, is obvious, we finally obtain the desired AP-equivalence between $\sharpcsp(\Delta_i,\FF)$ and $\sharpcsp(\FF)$.
\end{proofof}

%%%%%%%%%%%%%%%%%%%%%%%%

In Sections \ref{sec:basis-case-proof}--\ref{sec:general-case-proof},
we will concentrate on the proof of Proposition \ref{Delta-removal-arity-all}(2).
First, let $\FF$ denote any nonempty set of constraints. Obviously, $\sharpcsp(\FF)$ is AP-reducible to $\sharpcsp(\Delta_i,\FF)$ for every index $i\in\{0,1\}$. It therefore suffices to show the other direction (namely,  $\sharpcsp(\Delta_i,\FF)\APreduces \sharpcsp(\FF)$) for an appropriately chosen index $i$.  Hereafter, we suppose that $\FF$ is complement unstable,  and we choose a constraint $f$ in $\FF$ that is complement unstable. Furthermore, we assume that $f$ has arity $k$ ($k\geq1$). Our proof of Proposition \ref{Delta-removal-arity-all}(2) proceeds by induction on this index $k$.

%%%%%%%%%%%%%%%%%%%%%%%%%%%
\subsection{Basis Case: {\em k} = 1, 2}\label{sec:basis-case-proof}

Under the assumption described at the very end of Section \ref{sec:complement-stable}, we now target the basis case of
$k\in\{1,2\}$.  The induction case of $k\geq3$ will be discussed in Section \ref{sec:general-case-proof}. Notice that Condition ($a$) of Proposition \ref{Delta-removal-arity-all}(2) is necessary; to see this claim,
consider a constraint set $\FF=\{[1,0]\}$.

\s

(1) Assuming $k=1$, let $f=[x,y]$ with $x,y\in\algebraic$. Note that  $x\neq \pm y$. This is because, if $x= \pm y$, then $f$ has the form $x\cdot [1,\pm1]$ and $f$ must be complement stable, a contradiction. Hence,
it follows that $|x|\neq|y|$.
Henceforth, we wish to assert that $|x|>|y|$ (resp., $|x|<|y|$) leads to a conclusion that  $\Delta_0$ (resp., $\Delta_1$) is effectively T-constructible from $f$.  This assertion comes from the following simple observation.

\begin{claim}\label{algebraic-simul}
Let $x$ and $y$ be two arbitrary algebraic real numbers with $|x|>|y|$.
The constraint $\Delta_0 =[1,0]$ is effectively T-constructible from  $[1,y/x]$ via a p-convergence series $\Lambda=\{[1,(y/x)^{2n}] \mid n\in\nat^{+}\}$ for $\Delta_0$.
In the case of $\Delta_1$, a similar statement holds if $|x|<|y|$ (in place of $|x|>|y|$).
\end{claim}

\begin{proof}
Let us assume that $|x|>|y|$. We set $\lambda =y/x$ and define $g_m = [1,\lambda^{2m}]$ for every index $m\in\nat^{+}$. It is clear that the series $\Lambda = \{g_m\mid m\in\nat^{+}\}$ is indeed a p-convergence series for $\Delta_0 =[1,0]$. In addition, the definition of $g_m$ yields the effective T-constructibility of $\Delta_0$ from $[1,\lambda]$. The case of $|x|<|y|$ is similarly treated.
\end{proof}

Assuming $|x|>|y|$, let $\lambda=y/x$.
Claim \ref{algebraic-simul} implies that $\Delta_0 \econst [1,\lambda]$. Since $[1,\lambda]\econst f$, we derive by Lemma \ref{constructibility}(2)
the desired conclusion that $\Delta_0\econst f$. In the case of $|x|<|y|$, it suffices to define $\lambda=x/y$.

%%%%%%%%%
%%%%%%%%%

\s

(2) Assume that $k=2$ and let $f=(x,y,z,w)$ for certain numbers $x,y,z,w\in\algebraic$. For convenience, we will examine separately the following two cases: $|x|=|w|$ and $|x|\neq |w|$.

\s

[Case: $|x|=|w|$] We want to prove the following claim, which corresponds to Condition ($a$) of Proposition \ref{Delta-removal-arity-all}(2). Recall that $f$ is complement unstable.

\begin{claim}\label{k=2_|x|=|w|}
Assuming that $|x|=|w|$, both  $\Delta_0$ and $\Delta_1$ are effectively T-constructible from $f$.
\end{claim}

(a) Let us assume that $x=w$. Notice that $y\neq z$, because $y=z$ implies that $f$ is complement invariant, a contradiction. Since $y\neq z$, we set $g=f^{x_1=*}$, which equals $[x+z,x+y]$. Similarly, define $h=f^{x_2=*}=[x+y,x+z]$.  Note that the equation   $(x+y)^2=(x+z)^2$ is transformed into $(y-z)(2x+y+z)=0$, which is equivalent to  $2x+y+z=0$ since $y\neq z$.
If $|x+y|<|x+z|$, then we can effectively T-construct $[1,0]$ and $[0,1]$
from $g$ and $h$, respectively, as done in Case (1). Similarly, when
$|x+y|>|x+z|$, we obtain $[0,1]$ and $[1,0]$.
In the other case where $2x+y+z = 0$, we  start with a new constraint $f'=f^2$ (which equals $(x^2,y^2,z^2,w^2)$) in place of $f$.  Obviously,
$f'$ is effectively $T$-constructible from $f$.
Let us consider the simple case where $y^2\neq z^2$. Since $f'$ is not complement stable and $2x^2+y^2+z^2\neq0$, this case is reduced to the previous case. Finally, let us consider the case where $y^2=z^2$. Since $y\neq z$, we conclude that $y=-z$. {}From $2x+y+z=0$, instantly $x=0$ follows.
Thus, $f$ must equal $(0,y,-y,0)$, which is complement anti-invariant. This  contradicts our assumption.

\s

(b) Assume that $x=-w$. First, we claim that $y\neq -z$ because, otherwise, $f$ becomes complement anti-invariant. Let us consider a new constraint $f'=f^2$. If $y^2\neq z^2$, then $f'$ is not complement stable, and thus we can reduce this case to Case (a). Hence, it suffices to assume that $y^2=z^2$. This implies $y=z$ and we thus obtain $f=[x,y,-x]$.  Next, we define $g=f^{x_1=*}=[x+y,y-x]$. Note that $f^{x_1=*}=f^{x_2=*}$.
Consider the case where $|x+y|=|y-x|$. This is equivalent to  $xy=0$. If $x=y=0$, then $f$ is complement invariant. If $x=0$ and $y\neq0$, then $f$ is also complement invariant. If $x\neq0$ and $y=0$, then $f$ is complement anti-invariant. In any case, we obtain an obvious contradiction.

Finally, we deal with the case where $|x+y|\neq |y-x|$ (which is equivalent to $xy\neq 0$). Define $f'=f^{x_1=x_2}=[x,-x]$ and set $h(x_1,x_3)= \sum_{x_2\in\{0,1\}} f(x_1,x_2)f'(x_2) f(x_2,x_3)$, which equals $x\cdot [x^2-y^2,2xy,x^2-y^2]$. Note that $h\econst f$.

(i) If $x=y$, then we obtain $g=[2x,0]$. From this $g$, we can
effectively T-construct $[1,0]$. Now, consider another constraint
$h'(x_1)= \sum_{x_2} h(x_1,x_2) g(x_2)=[0,4x^4]$, from which we effectively T-construct $[0,1]$.

(ii) Assume that $x=-y$. Since $g=[0,-2x]$, $[0,1]$ is effectively T-constructible from $g$. Next, consider $h'$ defined as above. Obviously,  $h'=[4x^4,0]$ holds, and thus we effectively T-construct $[1,0]$ from $h'$.

(iii) Consider the case where $|x|\neq |y|$. If $|x+y|>|y-x|$ (equivalently, $xy>0$), then  we effectively T-construct $[1,0]$ from $g= [x+y,y-x]$
by Claim \ref{algebraic-simul}.
Since $f$ is of the form $[x,y,-x]$, when $|x|<|y|$, we take
a series $\Lambda = \{[(x/y)^{2i},1,(-x/y)^{2i}] \mid i\in\nat^{+}\}$, which is obviously a p-convergence series for $XOR=[0,1,0]$. Similar to Claim \ref{algebraic-simul}, the following assertion holds.

\begin{claim}\label{XOR-T-const}
Let $x,y$ be arbitrary algebraic real numbers with $|x|<|y|$. The constraint $XOR$ is effectively T-constructible from $[x/y,1,\pm x/y]$ via a p-convergence series $\Lambda = \{[(x/y)^{2i},1,(\pm x/y)^{2i}] \mid i\in\nat^{+}\}$ for $XOR$.
\end{claim}

Note that $[0,1]$ is effectively T-constructible from $\{XOR,[1,0]\}$. Since Claim \ref{XOR-T-const} yields $XOR \econst f$, we obtain $[0,1]\econst f$.

In the other case where $|x|>|y|$, $h$ has the form $\frac{1}{2y}\cdot [d,1,d]$, where $d=(x^2-y^2)/2xy$. If $x^2-y^2<2|xy|$, then a series $\Lambda =\{[d^{2i},1,d^{2i}] \mid i\in\nat^{+}\}$ is a p-convergence series for $XOR=[0,1,0]$ because of $|d|<1$.
By Claim \ref{XOR-T-const}, $XOR$ is effectively T-constructible from $[d,1,d]$. Since $[d,1,d]\econst f$, it follows that $XOR\econst f$.
From this result, the aforementioned relation $[0,1]\econst\{XOR,[1,0]\}$ implies the desired consequence that $[0,1]\econst f$.

In contrast, when $|x+y|<|y-x|$ (equivalently, $xy<0$), Claim \ref{algebraic-simul} helps us effectively T-construct $[0,1]$ from $g=[x+y,y-x]$. A similar argument as before shows that $[1,0]\econst f$ using $XOR$.

\s

[Case: $|x|\neq|w|$] In this case, it is enough to prove the following claim whose last part is equivalent to Condition ($b$).

\begin{claim}\label{k=2-and-|x|neq|w|}
Assume that $|x|\neq|w|$. There exists an index $i\in\{0,1\}$ satisfying $\Delta_i\econst f$. Moreover, if one of the following two conditions is satisfied, then both $\Delta_0$ and $\Delta_1$ are effectively T-constructible from $f$. The conditions include (i) $|x|>|w|$, and either $|x+y|<|z+w|$ or $|x+z|<|y+w|$, and (ii) $|x|<|w|$, and either $|x+y|>|z+w|$ or $|x+z|>|y+w|$.
\end{claim}

To show Claim \ref{k=2-and-|x|neq|w|}, let $g_0=f^{x_1=x_2}=[x,w]$, $g_1=f^{x_1=*}=[x+z,y+w]$, and $g_2=f^{x_2=*}=[x+y,z+w]$.
Clearly, $g_0,g_1,g_2\econst f$ holds.
In the case where $|x|>|w|$, we can effectively T-construct $\Delta_0$
from $g_0$ using Claim \ref{algebraic-simul}.
In addition, if either $|x+y|< |z+w|$ or $|x+z|< |y+w|$ holds,  we further
effectively T-construct $\Delta_1$ from either $g_1$ or $g_2$. Hence, we obtain both $\Delta_0$ and $\Delta_1$.  The case of $|x|<|w|$ is similarly treated.

%%%%%%%%%%%%
\subsection{Induction Case: {\em k} $\geq$ 3}\label{sec:general-case-proof}

As in the previous subsections, let $f=(z_1,z_2,\ldots,z_{2^k})$. We will deal with the remaining case of $k\geq3$.
In the next lemma,  from a given complement unstable constraint $f$ of arity $k$, we can effectively T-construct another arity-$(k-1)$ complement unstable constraint $g$ of a special form  that helps us apply an induction hypothesis.

\begin{lemma}\label{induction-step}
Let $k\geq3$ and let $f$ be any $k$-ary constraint. If $f$ is complement unstable, then there exists another constraint $g$ of arity $k-1$ for which (i) $g$ is complement unstable, (ii) $g\econst f$, and (iii) if $f$ satisfies one of Conditions ($a$)--($b$), then so does $g$.
\end{lemma}

We will briefly show that the induction case holds for Proposition \ref{Delta-removal-arity-all}(2), assuming that Lemma \ref{induction-step} is true.
By Lemma \ref{induction-step}, we take another arity-$(k-1)$ constraint $g$ in $\GG$ that satisfies Conditions (i)--(ii) of the lemma.
We then apply the induction hypothesis for Proposition \ref{Delta-removal-arity-all}(2) to conclude that either $\Delta_0$ or $\Delta_1$ is effectively T-constructible from $f$.
Next, we assume that $f$ violates one of Conditions ($a$)--($b$) given in Proposition \ref{Delta-removal-arity-all}(2).
If $f$ violates one of Conditions ($a$)--($b$), then the obtained constraint $g$  also violates one of those conditions. Hence, the induction hypothesis guarantees that both $\Delta_0$ and $\Delta_1$ are effectively T-constructible from $g$.

The above argument completes the induction case for Proposition \ref{Delta-removal-arity-all}(2).
Therefore, the remaining task of ours is to give the proof of Lemma \ref{induction-step}.

\begin{proofof}{Lemma \ref{induction-step}}
Let $f= (z_1,z_2,\ldots,z_{2^k})$. Since $f$ is complement unstable, there exists an appropriate index $\ell\in[2^k]$  satisfying   $z_{\ell}\neq 0$.
For each pair of indices $i,j\in[k]$ with $i<j$, we write $g^{(i,j)}$ to denote $f^{x_i=x_j}$ and then define
$\GG=\{g^{(i,j)}\mid 1\leq i<j \leq k\}$.
Note that  each constraint $g^{(i,j)}$ is effectively T-constructible
from $f$.

Let us begin with a simple observation.

\begin{claim}\label{fix-index}
Let $k\geq3$. For any index $j\in[2^k]$, there exist a constraint $g$ in $\GG$ and $k-1$ bits $a_1,a_2,\ldots,a_{k-1}$ satisfying  $z_j = g(a_1,a_2,\ldots,a_{k-1})$.
\end{claim}

\begin{proof}
Since $z_j$ is an output value of $f$, there exists an input tuple $(a'_1,a'_2,\ldots,a'_k)\in\{0,1\}^k$ for which $z_j = f(a'_1,a'_2,\ldots,a'_k)$. Since $k\geq3$, there are two indices $i,j\in[2^k]$ with $i<j$ satisfying $a'_i =a'_j$. For this special pair, it follows that $g^{(i,j)}(a'_1,\ldots,a'_{i-1},a'_{i+1},\ldots,a'_k) = f(a'_1,\ldots,a'_{i-1},a'_j,a'_{i+1},\ldots,a'_k) = z_j$. It therefore suffices to set the desired constraint $g$ to be $g^{(i,j)}$ and set $(a_1,a_2,\ldots,a_{k-1})$ to be $(a'_1,\ldots,a'_{i-1},a'_{i+1},\ldots,a'_k)$.
\end{proof}

%%%%%
%%%%%

Let us return to the proof of Lemma \ref{induction-step}.
First, we assume that $f$ satisfies Condition ($b$). Take
an index $j\in[2^{k-1}]-\{1\}$ for which either $(|z_1|-|z_{2^k}|)(|z_1+z_j| -  |z_{2^k-j+1}+z_{2^k}|)<0$ or $(|z_1|-|z_{2^k}|)(|z_1+z_{2^k-j+1}| -  |z_{j}+z_{2^k}|)<0$. By Claim \ref{fix-index}, we can choose a constraint $g\in\GG$ satisfying that $z_j=g(a_1,a_2,\ldots,a_{k-1})$ for a certain bit series $a_1,a_2,\ldots,a_{k-1}$. Note that this constraint $g$ also satisfies Condition ($b$) and is complement unstable. The lemma thus follows instantly.

Hereafter, we assume that $f$ does not satisfy Condition ($b$). When $\GG$ contains a complement unstable constraint, say, $g$, it has arity $k-1$ and $g\econst f$ holds.  Moreover, if $f$ further satisfies Condition ($a$), then $g$ also satisfies the condition because $g\in\GG$. We then obtain the lemma.
It therefore suffices to assume that
$\GG$ is complement stable.

Since $f$ is complement unstable,  either of the following two cases must occur. (1) There exists an index $i\in[2^{k-1}]$ satisfying $|z_i|\neq |z_{2^k-i+1}|$. (2) It holds that $|z_i|=|z_{2^k-i+1}|$ for every index $i\in[2^{k-1}]$, but  there are two distinct indices $i_0,j_0\in[2^{k-1}]$ for which  $z_{i_0} =z_{2^k-i_0+1}\neq0$ and $z_{j_0}= - z_{2^k-j_0+1}\neq0$.

\s

(1) In the first case, let us choose an index $i\in[2^{k-1}]$ satisfying  $|z_i|\neq |z_{2^k-i+1}|$. Claim \ref{fix-index} ensures the existence of  a constraint $g$ in $\GG$ such that $z_i = g(a_1,a_2,\ldots,a_{k-1})$ for  appropriately chosen $k-1$ bits $a_1,a_2,\ldots,a_{k-1}$. This implies that $z_{2^k-i+1} = g(a_1\oplus1,a_2\oplus1,\ldots,a_{k-1}\oplus1)$. By the choice of $i$, $g$ cannot be complement stable. Obviously, this is a contradiction against our assumption that $\GG$ is complement stable.

\s

(2) In the second case, let us take any two indices $i_0,j_0\in[2^{k-1}]$ satisfying  that $z_{i_0} =z_{2^k-i_{0}+1}\neq0$ and $z_{j_0}= - z_{2^k-j_{0}+1}\neq0$. We will examine two possible cases separately.

(i) Assume that a certain constraint $g\in \GG$ satisfies  both  $z_{i_0} = g(a_1,a_2,\ldots,a_{k-1})$ and $z_{j_0} = g(b_1,b_2,\ldots,b_{k-1})$
for appropriately chosen $2(k-1)$ bits $a_1,a_2,\ldots,a_{k-1},b_1,b_2,\ldots,b_{k-1}$. {}From the properties of $z_{i_0}$ and $z_{j_0}$, it follows that $g$ is complement unstable, and this fact clearly leads to a contradiction.

(ii) Finally, assume that Case (i) does not hold. This case is much more involved than Case (i).
By our assumption,   $|z_{i}| = |z_{2^k-i+1}|$ holds for all indices $i\in[2^{k-1}]$.  This assumption makes $f$ satisfy Condition ($a$).
To make the following argument simple, we will introduce several notations. First, we denote by $H'$ the set of all index pairs $(i,j)$ in $[2^k]\times[2^k]$  such that both $z_i = g(a_1,a_2,\ldots,a_{k-1})$ and $z_j = g(b_1,b_2,\ldots,b_{k-1})$ hold for  a certain constraint $g$ in $\GG$ and certain $2(k-1)$ bits  $a_1,a_2,\ldots,a_{k-1},b_1,b_2,\ldots,b_{k-1}$.
Notice that $(i,i)\in H'$ holds for every index $i\in[2^k]$.
Since Case (i) fails, we obtain $(i_0,j_0)\notin H'$.
Associated with the set $H'$, we define two new sets
$H = [2^k]\times[2^k] - H'$ and
$\hat{H} = \{i\in[2^k]\mid \exists j[(i,j)\in H]\}$.  Since $|z_i|=|z_{2^k-i+1}|$ for all $i\in[2^{k-1}]$,  $\hat{H}$ can be expressed as the disjoint union $\hat{H}_0\cup\hat{H}_{+}\cup\hat{H}_{-}$, where  $\hat{H}_0 = \{i\in\hat{H}\mid z_i=z_{2^k-i+1}=0\}$,  $\hat{H}_{+} = \{i\in \hat{H}\mid z_i = z_{2^k-i+1}\neq0\}$, and $\hat{H}_{-} = \{i\in\hat{H}\mid z_i = -z_{2^k-i+1}\neq0\}$. Concerning   $\hat{H}_{+}$ and $\hat{H}_{-}$, the following useful properties hold.

\begin{claim}\label{H+-and-H-}
Let $i,j\in[2^k]$ be any two indices with $i\leq 2^{k-1}$ and assume that $(i,j)\in H'$.
\begin{enumerate}\vs{-1}
\item If $j\in \hat{H}_{+}$ then $z_i = z_{2^k-i+1}$ holds.
\vs{-2}
\item If $j\in \hat{H}_{-}$ then $z_i = -z_{2^k-i+1}$ holds.
\end{enumerate}
\end{claim}

\begin{proof}
Assume that $z_i=0$. From $z_i =0$ follows $z_{2^k-i+1}=0$, because  $|z_i|=|z_{2^k-i+1}|$. We then obtain $z_i=\pm z_{2^k-i+1}$; thus,
the claim is trivially true. Henceforth, we consider the case where  $z_i\neq0$. Since $(i,j)\in H'$, there exist a constraint $g\in \GG$ and bits $a_1,a_2,\ldots,a_{k-1},b_1,b_2,\ldots,b_{k-1}$ for which  $z_i =g(a_1,a_2,\ldots,a_{k-1})$ and $z_j = g(b_1,b_2,\ldots,b_{k-1})$.
Recall that $g$ is complement stable by our assumption. In the case where $j\in \hat{H}_{+}$,  since $z_j = z_{2^k-j+1}\neq0$ holds,  $g$ must be complement invariant. Thus, for the index $i$, we obtain $z_i = z_{2^k-i+1}$. By a similar argument, when $j\in \hat{H}_{-}$, $g$ must be complement anti-invariant and thus $z_i = -z_{2^k-i+1}$ holds.
\end{proof}

Here, we claim that $\hat{H}_{+}$ and $\hat{H}_{-}$ are both nonempty. To see this claim, recall that the indices $i_0$ and $j_0$ satisfy $(i_0,j_0)\notin H'$, and thus they belong to $\hat{H}$; more specifically, it holds that $i_0\in\hat{H}_{+}$ and $j_0\in\hat{H}_{-}$ since $z_{i_0} = z_{2^k-i_0+1}$ and $z_{j_0} = -z_{2^k-j_0+1}$. By symmetry, we also conclude that $2^{k}-i_0+1\in\hat{H}_{+}$ and $2^{k}-j_0+1\in \hat{H}_{-}$.  In
Claim \ref{not-in-H-z_i}, we present another useful property of $\hat{H}$.

\begin{claim}\label{not-in-H-z_i}
For any index $i\in[2^k]$, if $i\notin \hat{H}$, then $z_i=0$ holds.
\end{claim}

\begin{proof}
Let $i$ be any index in $[2^k]-\hat{H}$. Without loss of generality, we assume that $i\leq 2^{k-1}$.
Toward a contradiction, we assume that $z_i\neq 0$.
As noted earlier,  $\hat{H}_{+}$ and $\hat{H}_{-}$ are nonempty.
Now, let us take two indices $j_1\in\hat{H}_{+}$ and $j_2\in\hat{H}_{-}$ and consider two pairs $(i,j_1)$ and $(i,j_2)$.
For the first pair $(i,j_1)$, if $(i,j_1)\notin H'$ holds, then $(i,j_1)$ must be in $H$, and thus $i$ belongs to $\hat{H}$. Since this is clearly a contradiction, $(i,j_1)\in H'$ follows. Similarly, we can obtain $(i,j_2)\in H'$ for the second pair $(i,j_2)$.
Claim \ref{H+-and-H-} then implies $z_i = z_{2^k-i+1}$ as well as $z_i = - z_{2^k-i+1}$. {}From these equations, $z_i=0$ follows. This is also a contradiction. Therefore, the claim is true.
\end{proof}

The rest of the proof proceeds by examining three cases, depending on the value of $k\geq3$.

\s

(a) Let us consider the base case of $k=3$ with $f=(z_1,z_2,\ldots,z_8)$.  By a straightforward calculation,  $H'$ is comprised of pairs $(2,3)$, $(2,4)$, $(2,5)$, $(3,4)$, $(3,5)$, $(3,7)$, $(4,6)$, $(4,7)$, $(5,6)$, $(5,7)$, $(6,7)$ and, moreover,  all pairs obtained from those listed pairs, say,  $(i,j)$ by exchanging two entries $i$ and $j$. Thus,
$\hat{H}$ equals $\{2,3,4,5,6,7\}$. Claim \ref{not-in-H-z_i} yields the equality $z_1=z_8=0$.
Now, we assume that $\hat{H}_0\neq\setempty$. In the case where $\hat{H}_0=\{4,5\}$, $\hat{H}_{+}$ is either $\{2,7\}$ or $\{3,6\}$ because $\hat{H}_{-}$ is nonempty. Now, we define
$h_2=f^{x_2=*}$, which equals $(z_3,z_2,z_7,z_6)$. Note that $h_2$ satisfies Condition ($a$).
If $h_2$ is complement stable, then it must hold  that either $z_i=z_{9-i}$ for all $i\in[4]$ or $z_i= -z_{9-i}$ for all $i\in[4]$. This implies that $f$ is complement stable, a contradiction. Therefore, $h_2$ is complement unstable.
The other cases ($\hat{H}_0=\{2,7\}$ and $\hat{H}_0=\{3,6\}$) are similar.

Next, assume that $\hat{H}_0=\setempty$.  Recall that $|\hat{H}_{+}|>0$ and $|\hat{H}_{-}|>0$ and note that $|\hat{H}_{+}|\neq |\hat{H}_{-}|$.
Let us consider the case where $|\hat{H}_{+}|>|\hat{H}_{-}|$. If $\hat{H}_{+}=\{2,4,5,7\}$ and $\hat{H}_{-}=\{3,6\}$, then we define $h_2=f^{x_2=*}$, which is $(z_3,z_2+z_4,z_5+z_7,z_6)$. Obviously, Condition ($a$) holds for $h_2$.  Since  $3\in\hat{H}_{-}$,  $z_3=-z_6$ holds; moreover, since $2,4\in\hat{H}_{+}$,  it follows that $z_2+z_4=z_5+z_7$.
We then conclude that  $h_2$ is not complement stable.
Similarly, if $\hat{H}_{+}=\{2,3,6,7\}$ and $\hat{H}_{-}=\{4,5\}$ (resp.,   $\hat{H}_{+}=\{3,4,5,6\}$ and $\hat{H}_{-}=\{2,7\}$), then consider $h_1=f^{x_1=*} = (z_5,z_2+z_6,z_3+z_7,z_4)$ (resp., $h_3=f^{x_3=*}=(z_2,z_3+z_4,z_5+z_6,z_7)$).  This constraint
$h_1$ is also complement unstable and satisfies Condition ($a$),
as requested.
The other case where $|\hat{H}_{-}|>|\hat{H}_{+}|$ is similarly treated.

\s

(b) Consider the case where $k=4$. It is not difficult to show
that $\hat{H} = \{4,6,7,10,11,13\}$.
Let us define $h = f^{x_1=*}$. This constraint $h=(w_1,w_2,\ldots,w_8)$ contains eight entries $w_i=z_{i}+z_{8+i}$ for all $i\in[8]$.  In particular,  $w_1=z_1+z_9$, $w_2=z_2+z_{10}$, $w_3=z_3+z_{11}$, $w_4=z_4+z_{12}$, $w_5=z_5+z_{13}$, $w_6=z_6+z_{14}$, $w_7=z_7+z_{15}$, and $w_8=z_8+z_{16}$. By Claim \ref{not-in-H-z_i}, it follows that $h=(0,z_{10},z_{11},z_{4},z_{13},z_{6},z_{7},0)$. Condition ($a$) is clearly met for this constraint $h$.
If $h$ is complement stable, either $z_{i}=z_{16-i+1}$ for all $i\in\{4,6,7\}$ or $z_{i}=-z_{16-i+1}$ for all $i\in\{4,6,7\}$.  This is impossible because $z_{i_0}=z_{16-i_0+1}\neq0$ and $z_{j_0}= -z_{16-j_0+1}\neq0$. Therefore, $h$ is complement unstable.

\s

(c) Assume that $k\geq5$. let us claim that $H=\setempty$. Assume otherwise. Let $(i,j)\in H$ and consider two $k$-bit series $a=a_1a_2\cdots a_k$ and $b=b_1b_2\cdots b_k$ satisfying that $z_i=f(a)$ and $z_j=f(b)$. Note that, for every distinct pair $s,t\in[k]$, $a_s=a_t$ implies $b_s\neq b_t$. For convenience, let $P_r=\{s\in[k]\mid a_s=r\}$ for each bit $r\in\{0,1\}$. Here, we examine only the case where $|P_0|\geq |P_1|$ since the other case is similar. Since $k\geq5$, it follows that $|P_0|\geq k/2\geq 3$; namely, there are at least three elements in $P_0$. For simplicity, let $1,2,3\in P_0$. Since $a_1=a_2=a_3=0$, there must be two distinct indices $i_1,i_2\in\{1,2,3\}$ for which $b_{i_1} = b_{i_2}$. Write $b^{(i_2)}$ for the $(k-1)$-bit series $b_1b_2\cdots b_{i_2-1}b_{i_2+1}\cdots b_k$. Similarly, we define $a^{(i_2)}$. By the choice of $(i_1,i_2)$, it holds that $z_i=f^{x_{i_1}=x_{i_2}}(a^{(i_2)})$ and $z_j=f^{x_{i_1}=x_{i_2}}(b^{(i_2)})$. This fact implies that $(i,j)\in H'$, a contradiction. Therefore,
we conclude that $H=\setempty$; that is, $H'=[2^k]\times[2^k]$. However, this contradicts our assumption that $(i_0,j_0)\notin H'$.

\s

This completes the proof of Lemma \ref{induction-step} and thus finishes the proof of Proposition \ref{Delta-removal-arity-all}(2).
\end{proofof}

Throughout the proof of Theorem \ref{key-Delta-elimination}, we have required the use of algebraic real numbers only in the proofs of Claims \ref{algebraic-simul} and \ref{XOR-T-const}. It is not known so far that the theorem  is still true for arbitrary real numbers.

%%%%%%%%%%%%%%%
%%%%%%%%%%%%%%%
\section{AP-Reductions without Auxiliary Unary Constraints}\label{sec:computability}

As a direct application of Theorem \ref{key-Delta-elimination},  we wish to prove our second theorem---Theorem \ref{main-theorem}---presented in Section \ref{sec:introduction}. To clarify the meaning of this theorem, we need to formalize the special constraints described in Section \ref{sec:introduction}. Let us introduce the following sets of constraints. Recall  that all constraints dealt with in this paper are assumed to output only {\em algebraic real values}.
\begin{enumerate}
\item Let $\DG$ denote the set of all constraints $f$ that are expressed by products of unary functions. A constraint in $\DG$ is called {\em degenerate}. When $f$ is symmetric, $f$  must  have one of the following three  forms:  $[x,0,\ldots,0]$, $[0,\ldots,0,x]$, and $y\cdot [1,z,z^2,\ldots,z^k]$ with $yz\neq 0$.
By restricting $\DG$, we define $\DG^{(-)}$ as the set of constraints of the forms $[x,0,\ldots,0]$, $[0,\ldots,0,x]$, $y\cdot[1,1,\ldots,1]$, and $y\cdot [1,-1,1,\ldots, -1\,\text{or}\,1]$, where $y\neq0$.
Naturally, both $\Delta_0$ and $\Delta_1$ belong to $\DG^{(-)}$.
\vs{-2}
\item The notation $\ED_{1}$ indicates the set of the following constraints:  $[x,\pm x]$,  $[x,0,\ldots,0,\pm x]$ of arity $\geq2$, and $[0,x,0]$ with $x\neq0$.
As a natural extension of $\ED_{1}$, let $\ED_{1}^{(+)}$ be composed of constraints  $[x,y]$,  $[x,0,\ldots,0,y]$ of arity $\geq2$, and $[0,x,0]$ with $x,y\neq0$. Notice that $[x,y]$ also belongs to $\DG$.
\vs{-2}
\item Let $\AZ$ be made up of all constraints of arity at least $3$ having
the forms $[0,x,0,x,\ldots,0\,\text{or}\,x]$ and  $[x,0,x,0,\ldots,x\,\text{or}\,0]$ with $x\neq0$.  The term ``$\AZ$'' indicates ``alternating zeros.'' Similarly, $\AZ_1$ denotes the set of all constraints of arity at least $3$ of the forms $[0,x,0,-x,0,x,0,\ldots,0\,\text{or}\, x\,\text{or}\,-x]$ and  $[x,0,-x,0,x,0\ldots,-x\,\text{or}\, x\,\text{or}\,0]$ with alternating $0$, $x$, and $-x$, where $x\neq0$.
\vs{-2}
\item The set $\BB_0$ consists of all constraints  $[z_0,z_1,\ldots,z_{k}]$ with $k\geq2$ and  $z_0\neq0$ that satisfy either (i) $z_{2i+1}=z_{2i+2}=(-1)^{i+1}z_0$ for all $i\in \nat$ satisfying $2i+1\in[k]$ or $2i+2\in[k]$,     or
    (ii) $z_{2i}=z_{2i+1}=(-1)^{i}z_0$ for all $i\in \nat$ with $2i\in[k]$ or $2i+1\in[k]$.
    As simple examples, $[1,1,-1]$ and $[1,-1,-1]$ belong to $\BB_0$.
\vs{-2}
\item Let $\OR$ denote the set of all constraints of the form $[0,x,y]$ with $x,y>0$. For instance, a special constraint $OR=[0,1,1]$ belongs to $\OR$.
\vs{-2}
\item Let $\NAND$ consist of all constraints of the form $[x,y,0]$ with $x,y>0$.
A Boolean constraint $NAND=[1,1,0]$ is in $\NAND$.
\vs{-2}
\item Let $\BB$ be comprised of all constraints of the form $[x,y,z]$ with $x,y,z>0$ and $xz\neq y^2$.
\end{enumerate}

In the presence of $\UU$, we obtain  (i) $\#\mathrm{SAT}\APreduces \sharpcsp(OR,\UU)$  \cite[Lemma 4.1]{Yam10a}, (ii)  $\sharpcsp(OR,\UU)\APreduces \sharpcsp(g,\UU)$ for every constraint $g$ in $\OR\cup\NAND$ \cite[Lemma 6.3]{Yam10a}, and (iii)  $\sharpcsp(OR,\UU)\APreduces \sharpcsp(g,\UU)$ for every constraint $g\in\BB$  \cite[Proposition 6.8]{Yam10a}. In conclusion, we can derive the following lemma.

\begin{lemma}\label{OR-and-B-Yam10a}
For every constraint $g$ in $\OR\cup\NAND\cup\BB$, it holds that $\#\mathrm{SAT}\APreduces \sharpcsp(g,\UU)$.
\end{lemma}

The first part of Theorem \ref{main-theorem} concerns the tractability of $\sharpcsp(\FF)$ when one of the two containments  $\FF\subseteq\DG\cup\ED_{1}^{(+)}$ and   $\FF\subseteq\DG^{(-)}\cup \ED_{1}\cup\AZ \cup\AZ_1\cup\BB_0$ holds.  For such a counting problem $\sharpcsp(\FF)$, it is already known in \cite[Theorem 5.2]{CLX09x}
and \cite[Theorem 1.2]{GGJ+10} that
$\sharpcsp(\FF)$ is solvable in polynomial time.

\begin{proposition}\label{compute-sig-set}{\rm \cite{CLX09x,GGJ+10}}\hs{1}
Let $\FF$ be any set of symmetric real-valued constraints. If  either  $\FF\subseteq \DG\cup\ED_{1}^{(+)}$ or $\FF\subseteq\DG^{(-)}\cup\ED_{1}\cup\AZ \cup\AZ_1\cup\BB_0$ holds, then $\sharpcsp(\FF)$ belongs to $\fp_{\algebraic}$.
\end{proposition}

%%%%%%%%
%%%%%%%%

Now, we come to the point of proving the second part of Theorem \ref{main-theorem}.
Let us first analyze the approximation complexity of $\sharpcsp(f)$ for an arbitrary symmetric constraint $f$ that are not included in $\DG\cup\ED_{1}^{(+)}\cup\AZ $. Note that, when $f$ is in $\DG\cup\ED_{1}^{(+)}\cup\AZ \cup\AZ_1\cup\BB_0$,  $\sharpcsp(f)$ belongs to $\fp_{\algebraic}$ by Proposition \ref{compute-sig-set}.
The following is a key claim required for the proof of Theorem \ref{main-theorem}.

\begin{lemma}\label{higher-case}
Let $f$ be any symmetric real-valued constraint of arity at least $2$. If $f\notin\DG\cup\ED_{1}^{(+)}\cup \AZ\cup \AZ_1\cup\BB_0 $, then, for any index $i\in\{0,1\}$,
there exists a constraint $g$ in $\OR\cup\NAND\cup\BB$ such that $g$ is effectively T-constructible from $\{f,\Delta_i\}$.
\end{lemma}

%%%%%%%%%%%%%%%%%%%%%%%
%%%%%%%%%%%%%%%%%%%%%%%

Theorem \ref{main-theorem} then follows, as shown below, by combining Theorem  \ref{key-Delta-elimination} and Proposition \ref{compute-sig-set} with an application of Lemma \ref{higher-case}.

\begin{proofof}{Theorem \ref{main-theorem}}
Since Proposition \ref{compute-sig-set} has already shown the first part of Theorem \ref{main-theorem}, we are focused on the last part of the theorem. To prove this part, we assume that $\FF\not\subseteq\DG\cup\ED_{1}^{(+)}$ and $\FF\not\subseteq\DG^{(-)}\cup \ED_{1}\cup\AZ \cup\AZ_1\cup\BB_0$.
Notice that, by our assumption,
$\FF$ should contain  a certain constraint whose entries are not all zero. Given a constraint set $\GG$, by applying Theorem \ref{key-Delta-elimination} to $\FF\cup\GG$, we obtain an index $i_0\in\{0,1\}$ for which  $\sharpcsp(\Delta_{i_0},\FF,\GG)\APequiv \sharpcsp(\FF,\GG)$.

If there exists a constraint $f$ in $\FF$ not in $\DG\cup\ED_{1}^{(+)}\cup \AZ \cup\AZ_1\cup\BB_0$, then we apply Lemma \ref{higher-case}
to obtain an appropriate constraint $g\in \OR\cup\NAND\cup\BB$ for which $g$ is effectively T-constructible from $\{f,\Delta_{i_0}\}$. By Lemma \ref{constructibility}(3), the theorem immediately follows.
Therefore, it is sufficient to consider the case where
$\FF\subseteq \DG\cup\ED_{1}^{(+)}\cup \AZ \cup\AZ_1\cup \BB_0$.
Now, let us choose two constraints $f_1,f_2\in\DG\cup\ED_{1}^{(+)}\cup\AZ \cup\AZ_1\cup\BB_0$ in $\FF$ for which $f_1\notin\DG^{(-)}\cup \ED_{1}\cup\AZ \cup\AZ_1\cup\BB_0$ and $f_2\notin\DG\cup\ED_{1}^{(+)}$. Note that $f_1\in \DG\cup\ED_1^{(+)}$ and $f_2\in \AZ\cup\AZ_1\cup\BB_0$.
Hereafter, we will prove the following claim.

\begin{claim}\label{OR-and-B}
There exists a constraint $g$ in $\OR\cup\NAND\cup\BB$ that is effectively T-constructible from $\{f_1,f_2,\Delta_{i_0}\}$.
\end{claim}

\begin{proof}
The proof of the claim proceeds as follows.
In general, $f_1$ has one of the following three forms: $[x,y]$, $[x,0,\ldots,0,y]$, and $x\cdot [1,z,z^2,\ldots,z^k]$ with $x,y\neq0$, $|x|\neq |y|$, $|z|\neq 1$, and $k\geq2$.
Notice that we can effectively T-construct $[x,y]$ from $[x,0,\ldots,0,y]$.
If $f_1$ is of the form $y\cdot [1,z,z^2,\ldots,z^k]$ with $y\neq0$, $|z|\neq 1$, and $k\geq2$, then no matter which of $\Delta_0$ and $\Delta_{1}$ is available, we can effectively T-construct $[1,z]$. Hence, we can
assume that $f_1$ has the form $[1,z]$.
In what follows, it suffices to assume that $f_1=[x,y]$ with $xy\neq0$ and $|x|\neq |y|$.

\s

(1) When $f_2\in \AZ$, there are two possibilities $f_2= u\cdot [0,1,0,1,\ldots,0\,\text{or}\,1]$ and  $f_2= u\cdot [1,0,1,0,\ldots,1\,\text{or}\,0]$ to occur. Since $f_2\notin \DG\cup \ED_1^{(+)}$, the arity $k$ of $f_2$ must be at least $3$. For simplicity, we set $u=1$. When $k>3$, we use the given constant constraint $\Delta_{i_0}$
to reduce $f_2$
to either $[1,0,1,0]$ or $[0,1,0,1]$. Here, let us  consider only the case where $f_2=[1,0,1,0]$ because the other case is similarly handled.

\sloppy
Here, we define two constraints $f(x_1,x_2,x_3) = f_1(x_1)f_2(x_1,x_2,x_3)$ and $h(x_1,x_2,x_3) = f(x_1,x_2,x_3)f(x_2,x_3,x_1)f(x_3,x_1,x_2)$. A simple calculation shows that
$h$ equals $x\cdot [x^2,0,y^2,0]$. Since $x^2\neq y^2$, we conclude that $h\notin \DG\cup\ED_{1}^{(+)}\cup\AZ \cup \AZ_1\cup\BB_0$.
Now, we apply Lemma \ref{higher-case} and
then obtain $g\econst \{h,\Delta_{i_0}\}$ for a certain constraint $g$ in $\OR\cup\NAND\cup \BB$. Since $h\econst \{f_1,f_2\}$, Lemma \ref{constructibility}(2) implies that $g\econst \{f_1,f_2,\Delta_{i_0}\}$.

\s

(2) If $f_2\in\AZ_1$, then we take $f_2^2$, which belongs to $\AZ$, and reduce this case to Case (1).

\s

(3) Assume that $f_2$ is in $\BB_0$. The constraint $f_2$ has either form $u\cdot [1,1,-1,\ldots,\pm1]$ or $u\cdot [1,-1,-1,\ldots,\pm1]$ with $u\neq0$. Apply $\Delta_{i_0}$ to $f_2$ appropriately and reduce $f_2$ to either $\pm u\cdot [1,1,-1]$ or $\pm u \cdot [1,-1,-1]$. Here, we assume that the constants ``$\pm u$'' to be $1$ for simplicity.  If $f_2=[1,1,-1]$, then we
define $h(x_1,x_2) = \sum_{x_3\in\{0,1\}} f_2(x_1,x_3)f_1(x_3)f_2(x_2,x_3)$. which equals $[x+y,x-y,x+y]$. Since $y\neq0$, it follows that $(x+y)^2-(x-y)^2=4xy\neq0$. Since $|x|\neq |y|$, $h$ belongs to $\BB$. If $f_2=[1,-1,-1]$, then $h$ equals $[x+y,y-x,x+y]$. A similar argument shows that $h$ also belongs to $\BB$. By the definition of $h$, it holds that $h\econst \{f_1,f_2,\Delta_{i_0}\}$.
\end{proof}

Let us return to the proof of the theorem. Let $\GG$ be any constraint set. By Claim \ref{OR-and-B}, an appropriately chosen constraint $g$ in $\OR\cup\NAND\cup\BB$ can be effectively T-constructible from $\{f_1,f_2,\Delta_{i_0}\}$. Lemma \ref{constructibility}(3) then implies that $\sharpcsp(g,\GG)\APreduces \sharpcsp(f_1,f_2,\Delta_{i_0},\GG)$. It follows from $f_1,f_2\in\FF$ that $\sharpcsp(f_1,f_2,\Delta_{i_0},\GG)\APreduces \sharpcsp(\FF,\Delta_{i_0},\GG)\APequiv \sharpcsp(\FF,\GG)$, where
the last AP-equivalence comes from the choice of $i_0$. Similarly, $g\econst \{f_1,f_2,\Delta_{i_0}\}$ implies $\sharpcsp(g,\GG)\APreduces \sharpcsp(f_1,f_2,\Delta_{i_0},\GG)$. Thus, combining all AP-reductions yields the desired consequence that $\sharpcsp(g,\GG)\APreduces \sharpcsp(\FF,\GG)$.
\end{proofof}

Briefly, we will describe how to prove Theorem \ref{dichotomy-theorem} even though a sketchy proof outline has been given in Section \ref{sec:introduction}. First, we consider the case where $\FF\subseteq \ED$. Since $\UU\subseteq \DG$, the problem $\sharpcsp(\ED,\UU)$
is AP-equivalent to $\sharpcsp(\DG,\ED^{(+)}_1)$ and, by Proposition \ref{compute-sig-set}, it is solvable in polynomial time.
On the contrary, assume that $\FF\nsubseteq \ED$.
Since $\DG\cup\ED_1^{(+)}\subseteq \ED$, it follows
that $\FF\nsubseteq \DG\cup\ED_1^{(+)}$.
Consider the first case where $\FF\nsubseteq \DG^{(-)} \cup\ED_1\cup \AZ\cup\AZ_1\cup\BB_0$.  Since $\FF\nsubseteq \DG\cup\ED_1^{(+)}$,
Theorem \ref{main-theorem} ensures the existence of a special constraint $g$ in $\OR\cup\NAND\cup\BB$ satisfying  $\sharpcsp(g,\UU)\APreduces \sharpcsp(\FF,\UU)$.
Using Lemma \ref{OR-and-B-Yam10a}, we obtain $\#\mathrm{SAT} \APreduces \sharpcsp(g,\UU)$.
By the choice of $g$, we easily conclude that $\#\mathrm{SAT}$ is AP-reducible to $\sharpcsp(\FF,\UU)$.
Next, we consider the second case where $\FF\subseteq \DG^{(-)} \cup\ED_1\cup \AZ\cup\AZ_1\cup\BB_0$. Since $\FF\nsubseteq \DG\cup \ED_{1}^{(+)}$, let us take a constraint $f_2$ from $\FF - (\DG\cup\ED_{1}^{(+)})$.
It follows that $f_2\in \AZ\cup\AZ_1\cup\BB_0$.
Next, we choose another constraint $f_1$ from $\UU-(\DG^{(-)}\cup\ED_1)$. It is clear by its definition that $f_1\notin \DG^{(-)}\cup \ED_1\cup\AZ\cup\AZ_1\cup\BB_0$.
Applying Claim \ref{OR-and-B}, we obtain a constraint
$g$ in $\OR\cup\NAND\cup\BB$ for which $g\econst \{f_1,f_2,\Delta_0,\Delta_1\}$. Lemma \ref{constructibility}(3) then implies that  $\sharpcsp(g,\UU)\APreduces \sharpcsp(f_1,f_2,\Delta_{0},\Delta_{1},\UU)$. Since $f_1,f_2\in\FF$ and $\Delta_0,\Delta_1\in\UU$, we conclude that $\sharpcsp(g,\UU)\APreduces \sharpcsp(\FF,\UU)$. Together with Lemma \ref{OR-and-B-Yam10a}, the desired AP-reduction $\#\mathrm{SAT} \APreduces \sharpcsp(\FF,\UU)$ follows.

%%%%%%%%%%%%%%%%%%%%
%%%%%%%%%%%%%%%%%%%%

To finish our entire argument, we still need to prove Lemma \ref{higher-case}.

\begin{proofof}{Lemma \ref{higher-case}}
Let $f$ be any symmetric real-valued constraint of arity $k\geq2$. For convenience, we write $\Gamma$ for $\DG\cup\ED_{1}^{(+)}\cup\AZ\cup \AZ_1\cup\BB_0$.
Throughout this proof, we assume that $f\notin\Gamma$ and that, for a fixed index $i_0\in\{0,1\}$,  $\Delta_{i_0}$ is given to use. Our proof proceeds by induction on $k$.

%%%%%%%%%%%%%
\paragraph{Case of $k=2$.}\label{sec:binary-case}

Let $f$ be any binary constraint not in $\Gamma$.
There are three major cases to consider separately, depending on the number of zeros in the output values of $f$.

\s

(B1) Consider the case where there are two zeros in the entries of $f$.
Obviously, $f$ must have one of the following three forms: $[x,0,0]$ ($\in \DG$),  $[0,0,x]$ ($\in\DG$), and $[0,x,0]$ ($\in\ED_{1}$) with $x\neq0$, yielding a contradiction against the assumption on $f$.

\s

(B2) Consider the case where there is exactly one zero in $f$. Note that $f$ must have one of the following forms: $[0,x,y]$, $[x,0,y]$, and $[x,y,0]$, where $xy\neq0$. For the first and the last forms, $f^2$ respectively belongs to $\OR$ and $\NAND$.
Moreover, $[x,0,y]$ is in $\ED_{1}^{(+)}$. The lemma thus follows.

\s

(B3)
Finally, consider the case where there is no entry of zero in $f=[x,y,z]$.
When $|xz|\neq y^2$, the constraint $f^2=[x^2,y^2,z^2]$ obviously
belongs to $\BB$; thus, it suffices to set $g$ in the lemma to be $f^2$.   If $xz=y^2$, then  $f$ has the form $[x,y,y^2/x]$ and thus it is in $\DG$. Finally, if $xz=-y^2$, then we obtain $f=[x,y,-y^2/x]$. Here, we wish to claim that $|x|\neq|y|\neq|-y^2/x|$, because this claim
establishes the membership of $f$ to $\BB$.
If $x=y$, then we obtain $-y^2/x=-x$ and thus $f=x\cdot[1,1,-1]$, which is obviously in $\BB_0$, a contradiction. When $x=-y$, we obtain $f=x\cdot [1,-1,-1]$, leading to another contradiction. Note that $y=y^2/x$ implies $x=y$, and $y=-y^2/x$ implies $y=-x$. Therefore, our claim is true.

%%%%%%%%%%%%%%%%%%%%%%%
\paragraph{Case of $k=3$.}\label{sec:single-sig-3}

We assume that $f$ has arity $3$. For convenience,
notations $x,y,z,w$ that will appear below as real values are assumed to be non-zero.

\s

(T1)
Consider the case where $f$ has exactly three zeros; that is,  $f$ is one of the following four forms: $[x,0,0,0]$, $[0,x,0,0]$, $[0,0,x,0]$, and $[0,0,0,x]$ with $x\neq0$. If $f\in\{[x,0,0,0],[0,0,0,x]\}$, then $f$ falls into $\DG$, a contradiction. Now, assume that $f=[0,x,0,0]$. We then define the desired constraint $g$ as $f^{x_1=*}$, which equals $x\cdot[1,1,0] = x\cdot NAND$, and thus it belongs to $\NAND$. Clearly, $g\econst f$ holds without use of $\Delta_{i_0}$.
The case of $f=[0,0,x,0]$ is treated similarly using $\OR$.

\s

\sloppy
(T2) Let us consider the case where $f$ has exactly two zeros; namely, $f$ has one of the forms: $[x,0,0,y]$, $[x,y,0,0]$, $[0,0,x,y]$, $[x,0,y,0]$, $[0,x,0,y]$, and $[0,x,y,0]$ with $xy\neq0$.
We exclude the case of $f=[x,0,0,y]$ because it belongs to $\ED_{1}^{(+)}$.  Remember that $\Delta_{i_0}$ is available to use.

(a) If $f=[x,y,0,0]$, then we define
$g=f^{x_1=x_2=x_3}$, which yields $x\cdot \Delta_0$. Since $x\neq0$, we can freely use $\Delta_0$. From the set $\{f,\Delta_0\}$, we can effectively T-construct a new constraint $h=[x,y,0]$. Thus, the constraint $h^2=[x^2,y^2,0]$, which is also effectively T-constructible, belongs to $\NAND$. The case $f=[0,0,x,y]$ is handled similarly using $\OR$ instead of $\NAND$.

(b) Assume that $f=[0,x,y,0]$. From $\{\Delta_{i_0},f\}$, we can effectively T-construct either $[0,x,y]$ or $[x,y,0]$, which is then reduced to Case (B2).

(c)
Finally, let $f=[x,0,y,0]$. In the case where $x=y$, we obtain $f= x\cdot [1,0,1,0]$, which belongs to $\AZ$. Moreover, if $x=-y$, then $f=x\cdot [1,0,-1,0]$ is in $\AZ_1$. In those cases, we clearly obtain a contradiction.
Hence, $|x|\neq |y|$ must hold. Let us consider another constraint  $g=f^{x_1=*}$, which equals $[x,y,y]$. Since $|xy|\neq y^2$, the constraint $g^2=[x^2,y^2,y^2]$ belongs to $\BB$. Since $g^2\econst f$, the lemma instantly follows.
The other case $f=[0,x,0,y]$ is similarly treated.

\s

(T3) Consider the case where $f$ has exactly one zero; that is,  $f=[x,y,z,0],[0,x,y,z],[x,y,0,z],[x,0,y,z]$ with $xyz\neq0$.

(a) If $f$ is of the form $[x,y,z,0]$, then we define $g=f^{x_1=x_2}$, which equals $(x,y,z,0)$. We then define $h(x_1,x_2) = g(x_1,x_2)^2g(x_2,x_1)^2$, implying $h=[x^4,y^2z^2,0]$. This constraint $h$ is in $\NAND$. By duality, we effectively T-construct $h'=[0,(xy)^2,z^4]$ from $[0,x,y,z]$. Clearly, $h'$ is a member of $\OR$.

(b) If $f=[x,0,y,z]$, then define $g(x_1,x_2) = \sum_{x_3,x_4\in\{0,1\}} f(x_1,x_3,x_4)f(x_2,x_3,x_4)$. This new constraint $g$ has the form $[A,B,C]$ with $A=x^2+y^2$, $B=yz$, and $C=2y^2+z^2$. Note that $A,C>0$ and $B\neq0$ because of $xyz\neq0$. We want to claim that $AC\neq B^2$. To show this inequality, assume that $AC=B^2$, that is, $(x^2+y^2)(2y^2+z^2)=y^2z^2$, or equivalently $2y^4+2x^2y^2+x^2z^2=0$. Since $z^2=-\frac{2y^2(x^2+y^2)}{x^2}$, it follows that $z^2<0$, a contradiction; hence, we obtain $AC\neq B^2$. Now, consider $g^2=[A^2,B^2,C^2]$. This constraint $g^2$ is clearly in
$\BB$. By duality, we can handle the case of $f=[x,y,0,z]$.

\s

(T4) Let us consider the case where $f$ has no zero; namely,  $f$ is of the form $[x,y,z,w]$ with
$xyzw\neq0$.
For the subsequent argument, we let $h_1(x_1,x_2) = \sum_{x_3,x_4\in\{0,1\}}f(x_1,x_3,x_4)f(x_2,x_3,x_4)$, $h_2(x_1,x_2) = \sum_{x_3\in\{0,1\}} f(x_1,x_3,x_3)f(x_2,x_3,x_3)$, and $h_3(x_1,x_3) = \sum_{x_2\in\{0,1\}}f(x_1,x_1,x_2) f(x_2,x_3,x_3)$.  Note that $h_1$, $h_2$, and $h_3$ are effectively T-constructible from $f$ alone.

(a)
If $|xz|\neq y^2$ and $|yw|\neq z^2$, then  define $h=f^{x_1=i_0}$ (i.e., $\sum_{x_1\in\{0,1\}} f(x_1,x_2,x_3)\Delta_{i_0}(x_1)$) and then obtain either $[x,y,z]$ or $[y,z,w]$. By our assumption, $h^2$ belongs to $\BB$.

(b) Now, assume that $|xz|\neq y^2$ and $|yw|=z^2$.

(i) We consider the first case where $yw=z^2$.
Without loss of generality, we can assume that $x=1$. Since $w=z^2/y$, $f$ equals $[1,y,z,z^2/y]$.
Note that $h_3$ is of the form $[A,B,C]$ with $A=1+y^2$, $B=z(z+1)$, and $C=z^2(1+z^2/y^2)$. Clearly, $A,C>0$ holds.
Now, let us study the case where $B\neq0$.
Since $y^2[AC-B^2]$ equals $z^2(y^2-z)^2$, it holds
that $AC=B^2$ iff $z=y^2$. Since $|z|\neq y^2$ by our assumption, we conclude that $AC\neq B^2$.
The constraint $h_3^2 =[A^2,B^2,C^2]$ therefore belongs to $\BB$.
Next, we consider the other case of $B=0$. Since $B=z(z+1)$, we immediately obtain $z=-1$; thus, $f$ must be of the form $[1,y,-1,1/y]$. Note that  $h_1=[A',B',C']$, where $A'=2(1+y^2)$, $B'=-(y+1/y)$, and $C'=2+y^2+1/y^2$. Obviously, $A',C'>0$ and $B'\neq0$ because $y\neq -1/y$.
The term $y^2[AC-B^2]$ is expressed as $y^2(y^2+1)^2(2y^2+1)$, which is obviously non-zero. We thus conclude that $AC\neq B^2$. Therefore, the constraint $h_1^2=[(A')^2,(B')^2,(C')^2]$ is in  $\BB$.
From $h_1^2\econst f$ and $h_3\econst f$, the lemma instantly follows.

(ii)
The second case is that $yw=-z^2$. As did before, we assume that $x=1$. Since $w=-z^2/y$,  $f$ equals $[1,y,z,-z^2/y]$. We then focus on the constraint  $h_3=[A,B,C]$, where $A=1+y^2$, $B=z(1-z)$, and $C=z^2(1+z^2/y^2)$.
Firstly, we examine the case of $B\neq 0$. Note that the value $y^2[AC-B^2]$ equals $z^2(y^2+z)$; thus, it holds that $AC=B^2$ iff $z=-y^2$. Since $|xz|\neq y^2$ and $A,C>0$, we conclude that  $|AC|\neq B^2$.
This places the constraint $h_3^2=[A^2,B^2,C^2]$ into $\BB$.
Secondly, we study the other case where $B=0$, or equivalently, $z=1$
because $B=z(1-z)$.
In this case, we use another constraint $h_2=[A',B',C']$, which is actually of the form $[2,y-1/y,y^2+1/y^2]$. If $|y|\neq 1$, then $B'\neq0$ holds.
Since $y^2[A'C'-(B')^2]$ has a non-zero value $(y^2+1)^2$,
the constraint $h_2^2$ is a member of $\BB$.
On the contrary, assume that $|y|=1$. When $y=1$,  $f$ equals $[1,1,1,-1]$, from which we obtain $h_1=[4,2,4]$. Obviously, $h_1$ belongs to  $\BB$. If $y=-1$, then we obtain $f=[1,-1,1,1]$, and thus $h_1=[4,-2,4]$ is also
in $\BB$.

(c) Consider the remaining case where $xz=\delta y^2$ and $yw=\delta' z^2$ for certain constants $\delta,\delta'\in\{\pm1\}$.

(i)  In the case where $xz=y^2$ and $yw=z^2$,  the constraint $f=[x,y,y^2/x,y^3/x^2]$ is obviously in $\DG$, a contradiction.

(ii) When $xz=y^2$ and $yw=-z^2$, we obtain $f=[x,y,y^2/x,-y^3/x^2]$.
Let us consider $h_2$. For simplicity, let $x=1$. The constraint $h_2$ therefore
has the form $[A,B,C]$, where $A=1+y^4$, $B=y-y^5$, and $C=y^2+y^6$. Note that $A,C>0$. It follows from $AC-B^2 = 4y^6$ that $AC$ is different from $B^2$.
Moreover, note that $B=0$ iff $y=1$.
Therefore, when $y\neq1$, the constraint $h_2^2=[A^2,B^2,C^2]$ belongs to
$\BB$. When $y=1$, $f$ has the form $[1,1,1,-1]$. The constraint $h_1$, which equals $[4,2,4]$, is clearly in $\BB$.
The proof for the case where $xz=-y^2$ and $yw=z^2$ is essentially the same.

(iii)
Finally, we consider the case where $xz=-y^2$ and $yw=-z^2$. This case
implies $f=[x,y,-y^2/x,-y^3/x^2]$.  By assuming $x=1$, we obtain $h_1=[A,B,C]$, where $A=1+2y^2+y^4$, $B=y(y^2-1)^2$, and $C=y^2+2y^4+y^6$. Obviously, $A,C>0$ holds. By a simple calculation, we obtain $AC-B^2 = 8y^4(y^4+1)$. It thus follows that $AC\neq B^2$.
Note that $B=0$ iff $y=\pm1$. If $|y|\neq1$, then the constraint
$h_1^2=[A^2,B^2,C^2]$ belongs to $\BB$. When $y=1$ and $y=-1$, we obtain $f=[1,1,-1,-1]$ and $f=[1,-1,-1,1]$, respectively, which are both in $\BB_0$, a contradiction.

%%%%%
\paragraph{Case of $k\geq4$.}

For convenience, let $u=f(0^k)$ and $w=f(1^k)$. There are four fundamental cases to examine, depending on the values of $u$ and $w$.
\s

[Case: $u=0$ and $w\neq0$]
Since the other case where $u\neq0$ and $w=0$ is symmetric, we omit that case.
First, we note that $g$ cannot belong to $\BB_0$ because $u=0$.
Now, let us consider the constraint
$h' = f^{x_1=x_2=\cdots =x_k} =[0,w]$.
Since $w\neq0$, from this constraint $h'$, we can effectively T-construct $\Delta_1=[0,1]$.
We then set  $g$ as $f^{x_1=1}$ (equivalently,  $\sum_{x_1\in\{0,1\}}f(x_1,\ldots,x_k)\Delta_1(x_1)$).
Since $g\notin\Gamma$ implies the desired consequence, in what follows, we assume that $g\in\Gamma$.

(a) In  the case where $g\in\DG$,  $g$ cannot be $[x,0,0,\ldots,0]$ because $w\neq0$. If $g=[0,0,\ldots,0,x]$, then $f$ must equal $[0,0,0,\ldots,0,x]$, which belongs to $\DG$, a contradiction. The remaining case is that $g$ has the form $x\cdot [1,y,y^2,\ldots,y^{k-1}]$ with $xy\neq0$.
Notice that $f=x\cdot [0,1,y,y^2,\ldots,y^{k-1}]$. If $y\neq -1$, then
we define $h=f^{x_1=*}$, which is $x\cdot[1,y+1,y(y+1),\ldots,y^{k-2}(y+1)]$. Since $y(y+1)\neq (y+1)^2$, $h$ is not in $\DG$. Because $y\neq0$, we obtain $y+1\neq1$; moreover, $y+1=y(y+1)$ implies $y+1\neq-1$. Therefore, $h$ cannot belong to $\BB_0$. In conclusion, $g$ is not in $\Gamma$.
We then apply the induction hypothesis to obtain the lemma. On the contrary, when $y=-1$, we consider another constraint $f^2=x^2\cdot[0,1,1,\ldots,1]$. This case can be reduced to the previous case of $y\neq-1$.

(b) Let us consider the case where $g$ belongs to $\ED_{1}^{(+)}$.
Since $k\geq4$, $g$ cannot have the form $[0,x,0]$.
If $g$ equals $[x,0,\ldots,0,w]$, then $f$ must be $[0,x,0,\ldots,0,w]$.
Define $h=f^{x_1=*}$, implying  $h=[x,x,0,\ldots,0,w]$. Since $h\notin\Gamma$, the induction hypothesis leads to the desired consequence.

(c) The next case to examine is that $g$ is in $\AZ$. The cases of $g=[x,0,x,0,\ldots,x,0]$ and $g=[0,x,0,x,\ldots,x,0]$ never occur
because of $w\neq0$.
If $g$ has the form $[x,0,x,0,\ldots,x]$, then $f$ must be of the form $[0,x,0,x,0,\ldots,x]$; thus, $f$ belongs to $\AZ$, a contradiction.  Moreover, if $g=[0,x,0,x,\ldots,0,x\,\text{or}\,0]$, then $f$ equals $[0,0,x,0,x,\ldots,0,x\,\text{or}\,0]$.  Let us define $h=f^{x_1=*}$, which equals $x\cdot [0,1,1,\ldots,1]$. Since  $h\notin\Gamma$,  we can apply the induction hypothesis to $h$.

(d) In the case of $g\in\AZ_1$, it holds that either $g=[x,0,-x,0,\ldots,w]$ or  $g= [0,x,0,-x,0,\ldots,w]$, where $w\in\{\pm x\}$. In the former case, $f$ is of the form $[0,x,0,-x,0,\ldots,w]$ and falls into $\AZ_1$, a contradiction.
In the latter case, we obtain $f=[0,0,x,0,-x,0,\ldots,w]$. Consider a new constraint $h=f^{x_1=*}$, which equals $[0,x,x,-x,-x,\ldots, w\,\text{or}\,w\pm x]$. Clearly, this does not belong to $\Gamma$, and thus we can apply the induction hypothesis to $h$.

\s

[Case: $uw\neq0$]
This case is split into three subcases: $|u|=|w|$, $|u|<|w|$, and $|u|>|w|$. The last subcase is essentially the same as the second one, we omit it.

(i) Let us assume that $|u|=|w|$. If $u=-w$, then the constraint $f'=f^2$
satisfies $f'(0,\ldots,0)=f'(1,\ldots,1)$, and thus it falls into the case of $u=w$. Henceforth, we will discuss only the case of $u=w$.
Here, we assume that $\Delta_0$ is available for free (\ie $i_0=0$). The  other case $i_0=1$ is similar. Note that the constraint $g=f^{x_1=0}$ is effectively T-constructible from $\{f,\Delta_0\}$. When $g\notin \Gamma$, the induction hypothesis can be directly applied to $g$; therefore, we assume below that $g$ is a member of $\Gamma$.

(a) We start with the case where $g\in\DG$. When $g=[u,0,\ldots,0]$, $f$ has the form $[u,0,\ldots,0,u]$. This constraint $f$ thus belongs to $\ED_{1}$, a contradiction against $f\notin\ED_1^{(+)}$.
When $g=[0,0,\ldots,0,x]$ with $x\neq0$, since $f=[0,0,\ldots,0,x,u]$, we can effectively T-construct $\Delta_1=[0,1]$ from $f$. We then define $h=f^{x_1=*}$ of the form $[0,0,\ldots,x,x+u]$. If $x\neq -u$, then the constraint $h'= h^{x_1=x_2=\cdots =x_{k-3}=1}$ has the form $[0,x,x+u]$, and thus its induced constraint $(h')^2$ belongs to $\OR$. Otherwise, we consider $f'=f^2=[0,0,\ldots,u^2,u^2]$ and reduce this case to the previous
case of $x\neq-u$.
Next, assume that $g= x\cdot[1,z,z^2,\ldots,z^{k-1}]$ for $x,z\neq0$.
In what follows, when $z=-1$, we use $f'=f^2$ instead of $f$.
If $z\neq -1$, then $f$ is of the form $x\cdot [1,z,z^2,\ldots,z^{k-1},u/x]$. Since $f\notin \DG$, it follows that $u/x\neq z^k$ (equivalently, $u\neq xz^{k}$). Let us consider $h=f^{x_1=*}$ of the form $x(1+z)\cdot [1,z,z^2,\ldots,z^{k-2},A]$, where $A=\frac{xz^{k-1}+u}{x(1+z)}$. Notice that $x=u/z^{k}$ iff $A= z^{k-1}$. This equivalence leads to $A\neq z^{k-1}$. Therefore, $h$ does not belong to $\Gamma$. We simply apply the induction hypothesis to $h$.

(b) Assume that $g\in\ED_{1}^{(+)}$. Since $g$ equals $[u,0,\ldots,0,x]$,  $f$ has the form $[u,0,\ldots,0,x,u]$. Let us define $h=f^{x_1=*}$, which equals $[u,0,\ldots,x,u+x]$. When $u\neq-x$, we apply the induction hypothesis to $h$ because of $h\notin\Gamma$.  On the contrary, if $u=-x$, then we consider another constraint $f^2$ ($=[u^2,0,\ldots,0,x^2,u^2]$) and reduce this case to the previous case of $u\neq-x$.

(c) Next, we assume that $g$ is in $\AZ$. Note that $g$ cannot have the form  $g=[u,0,u,0,\ldots,0]$ because, otherwise, $f$ equals $[u,0,u,0,\ldots,0,u]\in\AZ$, a contradiction.  When $g$ is of the form $[u,0,u,0,\ldots,0,u]$,   we easily obtain $[u,2u]$ from the constraint $h=f^{x_1=*}=[u,u,\ldots,u,2u]$. By Claim \ref{algebraic-simul}, we effectively T-construct $\Delta_1=[0,1]$ from $[u,2u]$.
Now, we apply $\Delta_1$ repeatedly to $f$ and then obtain $g'=[0,u,u]= u\cdot OR$, which obviously belongs to $\OR$.

(d) When $g$ is in $\AZ_1$, $g$ is not of the form $[u,0,-u,0,\ldots,-u,0]$ since, otherwise, $f$ is $[u,0,-u,\ldots,-u,0,u]$ and it belongs to $\AZ_1$, a contradiction. If $g$ equals $[u,0,-u,\ldots,0,\pm u]$, then $f$ has the form $[u,0,-u,\ldots,0,\pm u,u]$. We then take $f^2=[u^2,0,u^2,\ldots,0,u^2,u^2]$ and reduce this case to Case (c). Next, assume that $g=[u,0,-u,\ldots,u,0]$. Since $f=[u,0,-u,\ldots,u,0,u]$, the constraint $f^{x_1=*,x_2=*}$ equals $[0,-2u,0,\ldots,0,2u,2u]$, from which we immediately obtain $[0,2u] = 2u\cdot \Delta_{1}$. Using this $\Delta_{1}$, we can effectively T-construct $[0,2u,2u] = 2u\cdot OR$,
which obviously is in $\OR$.

(e) Assuming that $g\in \BB_0$, let $g=[u,-u,-u,u,\ldots,z_{k-1},z_k]$, where $z_{k-1},z_k\in\{u,-u\}$. As the first case, consider the case where   $(z_{k-1},z_k)=(-u,u)$; namely, $f$ has the form $u\cdot [1,-1,-1,1,\ldots,-1,1,1]$. Obviously, $f$ belongs to $\BB_0$, a  contradiction. A similar argument works for the case of $(z_{k-1},z_k)=(-u,-u)$. Now, assume that $(z_{k-1},z_k)=(u,u)$. Since $f= u\cdot[1,-1,-1,\ldots,1,1,1]$, the constraint $f^{x_1=*}$ must be
of the form $u\cdot[0,-2,0,\ldots,0,2,2]$. From this constraint, we effectively T-construct $\Delta_1=[0,1]$. Hence, $[0,2,2]$ is also  effectively T-constructed  from $\{f,\Delta_1\}$.
The resulted constraint $[0,2,2]$ is clearly in $\OR$.
Finally, consider the case of $(z_{k-1},z_k)=(u,-u)$. Since $f= u\cdot[1,-1,-1,\ldots,1,-1,1]$, the constraint $f^{x_1=*,x_2=*,\ldots,x_{k-1}=*}$ has the form $[0,y]$ for a
certain non-zero value $y$.
This allows us to use $\Delta_1$ freely. Applying $\Delta_1$ repeatedly, we easily obtain $[-2,-2,0]$ from $f^{x_1=*,x_2=*}= u\cdot [-2,-2,\ldots,0]$. This constraint $[-2,-2,0]=-2\cdot NAND$ belongs to $\NAND$.

(ii) Let us consider the second case where $|u|<|w|$.  Here, we define $h'(x_1)=f(x_1,x_1,\ldots,x_1)$. Since  $h'=[u,w]$ holds, we effectively T-construct $\Delta_1=[0,1]$ from $h'$ by Claim \ref{algebraic-simul}. As a result,  from $\{f,\Delta_1\}$, we further effectively T-construct  two constraints $g=f^{x_1=1}$  and $h=f^{x_1=*}$.
In the case of $g\notin\Gamma$, the desired result follows from the induction hypothesis. Hereafter, we assume that $g$ is indeed in $\Gamma$.

(a) If $g\in\DG$, then $g$ must be either  $[0,0,\ldots,0,w]$ or $x\cdot [1,z,z^2,\ldots,z^{k-1}]$. In the former case,
$f$ equals $[u,0,0,\ldots,0,w]$, and thus it belongs to $\ED_{1}^{(+)}$. This is a clear contradiction.
In the latter case, $f$ has the form $x\cdot [u/x,1,z,z^2,\ldots,z^{k-1}]$.

If $x=uz$, then $f$ equals $\frac{x}{z}\cdot [1,z,z^2,\ldots,z^{k-1}]$, which belongs to $\DG$, a contradiction.
Therefore, we obtain $x\neq uz$. Now assume that $z=-1$.
Since $x\neq uz$, $x+u\neq0$ holds. Note that the constraint
$h$ is of the form $x\cdot [\frac{u+x}{x},0,0,\ldots,0]$, from which
we can effectively T-construct $[\frac{u+x}{x},0]$. This unary constraint allows us to use  $\Delta_0=[1,0]$ freely. An application of $\Delta_0$ to $f$ generates $g'=x\cdot [u/x,1,z,z^2,\ldots,z^{k-2}]$. Since $g'\notin\Gamma$, we can apply the induction hypothesis.
Next, assume that $z\neq-1$. By a simple calculation, we obtain
$h = x(1+z)\cdot [\frac{u+x}{x(1+z)},1,z,\ldots,z^{k-2}]$. It follows from $x\neq uz$ that $\frac{u+x}{x(1+z)} \neq \frac{1}{z}$.   Hence, $h$ is not in $\Gamma$. The induction hypothesis then leads to the desired consequence.

(b) Consider the case of $g\in\ED_{1}^{(+)}$. Let $g$ be  $[x,0,\ldots,0,w]$ for a certain constant $x\neq0$; thus, $f$ equals
$[u,x,0,\ldots,0,w]$. If $u\neq-x$, then the constraint $h$ has the form $[u+x,x,0,\ldots,0,w]$. Since $u+x\neq0$, $h$ does not belong to $\Gamma$. We then apply the induction hypothesis to $h$. On the contrary, when $u=-x$, we instead substitute $f^2=[u^2,x^2,0,\ldots,0,w^2]$  for  $f$ and make this case reduced to the previous case of $u\neq-x$.

(c) Let us consider the case where $g$ is in $\AZ$. Since $g$ cannot have the form $[0,x,0,x,\ldots,x,0]$, we first assume that $g=[0,x,0,x,\ldots,x]$. The original constraint $f$ then has the form $[u,0,x,0,x,\ldots,x]$
with $x=w$.   Notice that $x\neq u$ because $|u|<|w|=|x|$.
Since $0<|u|<|x|$, the constraint $h=[u,x,x,\ldots,x]$ does not belong to $\Gamma$. The induction hypothesis can be applied to $h$.
In the case of $g=[x,0,x,0,\ldots,x]$, on the contrary, we consider $g'=(f^2)^{x_1=*}$, which is $[u^2+x^2,x^2,\ldots,x^2]$. From this $g'$, we obtain $[u^2+x^2,x^2]$. Claim \ref{algebraic-simul} again allows us to use $\Delta_0 =[1,0]$. Apply $\Delta_0$ repeatedly to the constraint $f^2$. We then obtain $[u^2,x^2,0]$, which clearly belongs to $\NAND$.

(d) Assume that $g$ belongs to $\AZ_1$.  Note that $g=[0,x,0,-x,0,\ldots,0]$
and $g=[x,0,-x,0,x,\ldots,0]$ are both impossible. First, we assume that  $g=[0,x,0,-x,\ldots,\pm x]$; thus, $f$ has the form $[u,0,x,0,-x,\ldots,\pm x]$, where $w=\pm x$. Since $|u|<|w|=|x|$, the constraint $h$ ($=[u,x,x,-x,-x,\ldots,\pm x]$) cannot belong to $\Gamma$. We then apply the induction hypothesis to $h$. Next, assume that $g=[x,0,-x,0,x,\ldots,\pm x]$. Since $f$ has the form $[u,x,0,-x,0,x,\ldots,\pm x]$, we consider a new constraint $g'=(f^2)^{x_1=*}$, which equals $[u^2+x^2,x^2,x^2,\ldots,x^2]$. From this constraint $g'$, we obtain $[u^2+x^2,x^2]$. The constant unary constraint $\Delta_0=[1,0]$ can be effectively T-constructed from $[u^2+x^2,x^2]$ by Claim \ref{algebraic-simul}. Using this $\Delta_0$, we obtain from $f^2$ the constraint $[u^2,x^2,0]$, which belongs to $\NAND$.

(e) Assuming that $g\in \BB_0$, we first consider the case where   $g=[x,-x,-x,x,\ldots,\pm x]$ with $x = \pm w$.
Note that $f$ must have the form $x\cdot [u/x,-1,-1,1,\ldots,-1,\pm1]$. Since $f\notin\AZ_1$, $u\neq x$ (equivalently, $u+x\neq 2x$) follows. If $u+x\neq -2x$, then the constraint  $h = x\cdot [(u+x)/x,-2,0,2,0,\cdots,0\,\text{or}\,\pm2]$ cannot belong to $\Gamma$, and thus we can apply the induction hypothesis to $h$.
The remaining case is $u+x=-2x$ (equivalently, $u=-3x$). In this case, we obtain $f= x\cdot [-3,1,-1,-1,1,\cdots,\pm1]$. Since $f^{x_1=x_2=\cdots =x_{k}}$ is $[-3,\pm1]$, we can effectively T-construct $\Delta_0=[1,0]$ by Claim \ref{algebraic-simul}. Applying $\Delta_0$ repeatedly to $f$, we obtain a new constraint $h'=[-3,1,-1]$, which clearly belongs to $\BB$.
The case where $g=[x,x,-x,-x,\ldots,\pm x]$ can be similarly handled.

\s

[Case: $u=w=0$]
Here, we assume that $\Delta_{i_0}=\Delta_0$. The other case of $\Delta_{i_0}= \Delta_1$ is similarly handled.
Now, we effectively T-construct $g=f^{x_1=0}$ from $\{f,\Delta_0\}$. Note that $g(0^{k-1})=0$. As done before, it suffices to consider the case where
$g$ is in $\Gamma$.
Clearly, $g\notin\ED_1^{(+)}\cup\BB_0$, and thus $g$ must be in $\DG\cup \AZ\cup\AZ_1$.
In the following argument, $h$ refers to $f^{x_1=*}$.

(a) Assume that $g$ is in $\DG$. Since $g(0^{k-1})=0$, $g$ is of the form $[0,0,\ldots,0,x]$ with $x\neq0$.  Since $f$ equals $[0,0,\ldots,0,x,0]$,
$h$ coincides with $[0,0,\ldots,x,x]$. Obviously,  $h\notin\Gamma$, and thus the induction hypothesis can be applied to $h$.

(b) When $g$ belongs to $\AZ$,  $g$ must have the form $[0,x,0,x,\ldots,0]$, because $g=[0,x,0,x,\ldots,x]$ implies $f=[0,x,0,x,\ldots,x,0]\in\AZ$, which leads to  a contradiction. Since  $f=[0,x,0,x,\ldots,0,0]$, we obtain $h=[x,x,x,\ldots,x,0]$. Since $h\notin\Gamma$, the induction hypothesis then leads to the desired consequence.

(c) Assume that $g\in\AZ_1$. If $g=[0,x,0,-x,\ldots,\pm x]$, then the original constraint $f=[0,x,0,-x,\ldots,\pm x,0]$ is already in $\AZ_1$, a contradiction. Hence, $g$ must have the form $[0,x,0,-x,\ldots,0]$, yielding   $h=[x,x,-x,\ldots,0]$. Clearly, $h$ does not belong to $\Gamma$. Finally, we  apply the induction hypothesis to $h$.
\end{proofof}

%%%%%%%%%%%%%%%%%%%%%%%%%
%%%%%%%%%%%%%%%%%%%%%%%%%
\subsection*{Appendix: Proof of Lemma \ref{constructibility}}

In what follows, we will give the missing proof of Lemma \ref{constructibility}. For any constraint $f$ of arity $k$, the notation $\max|f|$ indicates the maximum value $|f(x)|$ over all inputs  $x\in\{0,1\}^k$.

\s

(1)--(2)  These properties (reflexivity and transitivity) directly come from the definition of effective T-constructibility.

\s

(3) Let $(\HH_1,\HH_2,\ldots,\HH_n)$ be a generating series of $\FF_1$ from $\FF_2$. We need to show that $\sharpcsp(\HH_i,\GG) \APreduces \sharpcsp(\HH_{i+1},\GG)$  for each adjacent pair $(\HH_i,\HH_{i+1})$, where $i\in[n-1]$. By Lemma \ref{AP-property}, $\APreduces$ is transitive; thus,
it follows that $\sharpcsp(\HH_1,\GG)\APreduces \sharpcsp(\HH_n,\GG)$.
This is
clearly equivalent to $\sharpcsp(\FF_1,\GG)\APreduces \sharpcsp(\FF_2,\GG)$, as requested.

Taking an arbitrary pair $(\HH_i,\HH_{i+1})$ with $i\in[n-1]$,
we treat the first case where $(\HH_i,\HH_{i+1})$ satisfies Clause (I) of Definition \ref{def:constructibility}. Consider a constraint frame $\Omega=(G,X|\HH',\pi)$ with $\HH'\subseteq \HH_i\cup\GG$. For convenience, let $\HH'=\{f_1,f_2,\ldots,f_d\}$. Take $f_i$ inductively and consider all  subgraphs of $G$ that represents $f_i$. Choose such subgraphs one by one. Now, let $G_{f_i}$ be such a subgraph. By Clause (I), there exists another finite graph $G'$ that realizes $f_i$ by $\HH_{i+1}$. We replace $G_{f_i}$ in $G$ by $G'_{f_i}$. After all the subgraphs representing $f_i$ are replaced, the obtained graph, say, $G'$ constitutes a new constraint frame $\Omega'$. It is not difficult to show that $csp_{\Omega'}$ equals $csp_{\Omega}$. We continue this replacement process for all $f_i$'s. In the end, we conclude that $\sharpcsp(\HH_i,\GG)\APreduces \sharpcsp(\HH_{i+1},\GG)$.

We will examine the second case where  $(\HH_i,\HH_{i+1})$ satisfies Clause (II) of Definition \ref{def:constructibility}. In what follows, for ease of
our argument, we assume that $\HH_i=\{f\}$ and we want to claim that $\sharpcsp(f,\GG)\APreduces \sharpcsp(\HH_{i+1},\GG)$. Take a p-convergence series $\Lambda$ for $f$, which is effectively T-constructible from $\HH_{i+1}$.
Our claim is split into two parts: (a) $\sharpcsp(f,\GG)\APreduces \sharpcsp(\Lambda,\GG)$ and (b) $\sharpcsp(\Lambda,\GG)\APreduces \sharpcsp(\HH_{i+1},\GG)$.
We will prove these parts separately. Since (b) is easy, we start with (b).

\s

\sloppy
(b) We intend to show that $\sharpcsp(\Lambda,\GG)\APreduces \sharpcsp(\HH_{i+1},\GG)$. Let $\Lambda =(f_1,f_2,\ldots)$ and  $\HH_{i+1}=\{g_1,g_2,\ldots,g_d\}$. Now, we take any constraint frame $\Omega = (G,X|\GG',\pi)$ with $\GG'\subseteq \Lambda\cup\GG$, given to $\sharpcsp(\Lambda,\GG)$. Since the constraint set $\GG'$ is finite, for simplicity, we assume that $\GG'$ is composed of constraints $h_1,h_2,\ldots,h_s,f_{i_1},f_{i_2},\ldots,f_{i_t}$, where $s\in\nat$, $t\in\nat^{+}$, and each constraint $h_i$ belongs to $\FF-\Lambda$. For this  constraint frame $\Omega$, we will explain how to compute the value $csp_{\Omega}$. Since $\Lambda$ is effectively T-constructible  from $\GG$,  there exists  a polynomial-time DTM $M$ that, for each index $j\in[t]$, generates an appropriate graph $\tilde{G}_{i_j}$ realizing $f_{i_j}$ from any graph $G_{i_j}$ representing $f_{i_j}$.

Each node $v$ labelled $f_{i_j}$ ($j\in[t]$) in $G$ corresponds to a unique subgraph $G_{i_j}$, including all dangling edges adjacent to $v$, that  represents $f_{i_j}$.
By running $M$ on $G_{i_j}$,  we obtain another subgraph $\tilde{G}_{i_j}$ realizing $f_{i_j}$, which contains all the dangling edges of $G_{i_j}$.  It is therefore possible to generate from $G$ another bipartite graph $\tilde{G}$ in which every subgraph $G_{i_j}$ of $G$ representing $f_{i_j}$ is replaced by its associated subgraph $\tilde{G}_{i_j}$ obtained from $G_{i_j}$ by $M$.  We denote by $\Omega'$ the constraint frame obtained from $\Omega$ by replacing $G$ with $\tilde{G}$ and by modifying $\pi$ accordingly.  The definition of ``realizability'' implies that $csp_{\Omega'}$ equals $\gamma\cdot csp_{\Omega}$ for an appropriate number $\gamma\in\algebraic$.  Since $\tilde{G}$ contains only constraints in $\HH_{i+1}\cup\GG$,  $\Omega'$ must be a valid input instance to $\sharpcsp(\HH_{i+1},\GG)$. As a result,
we conclude that $\sharpcsp(\Lambda,\GG)$ is AP-reducible to $\sharpcsp(\HH_{i+1},\GG)$.

\s

(a) We want to claim that $\sharpcsp(f,\GG)\APreduces \sharpcsp(\Lambda,\GG)$. This claim is proven by modifying the proof of
\cite[Lemma 9.2]{Yam10a}.
Hereafter, assume that $f$ is of arity $k$ and let $\Lambda=(g_1,g_2,\ldots)$.  For convenience, we define $AC=\{x\in\{0,1\}^k\mid f(x)\neq0\}$. By Eq.(\ref{eqn:convergence}), there exists a constant $\lambda\in(0,1)$ such that, for every $m\in\nat^{+}$ and every  $x\in\{0,1\}^k$, certain constants  $c,d\in\{\pm1\}$ satisfies the following condition:
\begin{quote}
\begin{itemize}
\item[(*)\;\;] $(1+\lambda^m c)g_m(x) \leq f(x) \leq (1+\lambda^m d)g_m(x)$ for all $x\in AC$, and $|g_m(x)|\leq \lambda^m$ for all $x\in\{0,1\}^k-AC$.
\end{itemize}
\end{quote}
Without loss of generality, we can assume that $\gamma$ is an algebraic real number.

Let us take  any constraint frame $\Omega=(G,X|\GG',\pi)$ with $G=(V_1|V_2,E)$ and $\GG'\subseteq \{f\}\cup\GG$ given as an input instance to $\sharpcsp(f,\GG)$. It is enough to consider the case where $f$ appears
in $\GG'$.  Let $p_f$ denote the total number of nodes in $V_2$ whose
labels are $f$.
For simplicity, write $L$ for the set of all $2^k$-tuples $\ell=(\ell_{x_1},\ell_{x_2},\ldots,\ell_{x_{2^k}})\in\nat^{2^k}$ satisfying that $\sum_{i\in[2^k]}\ell_{x_i}= p_f$, where each $x_i$ denotes the  lexicographically $i$th string in $\{0,1\}^k$. In addition, we set $L_f =\{\ell\in L\mid \forall i\in[2^k]
\;[\  f(x_i)=0\rightarrow \ell_{x_i}=0\ ]\}$.
It is not difficult to show by Eq.(\ref{eqn:csp-def})
that $csp_{\Omega}$ can be expressed in   the form
$\sum_{\ell\in L_f} \alpha_{\ell}(\prod_{x\in AC}f(x)^{\ell_x})$ for appropriately chosen numbers $\alpha_{\ell}\in \algebraic$, provided that $0^0$ is treated as $1$ for technical reason.

We set  $a_0=2^k!\,2^{4k}$ and $b_0= [1+ (2\max|f|)^{|V_2|}]\cdot \sum_{\ell\in L-L_f}|\alpha_{\ell}|$, which are obviously independent of $m$. Meanwhile, we arbitrarily fix an integer $m\in\nat^{+}$ that satisfies both $\lambda^ma_0<1$ and $\lambda^mb_0<1$, and we denote by $\Omega_m$  the constraint frame obtained from $\Omega$ by replacing every node labeled $f$ with a new node having the label $g_m$.
Concerning this $\Omega_m$, its value $csp_{\Omega_m}$ coincides with the sum  $\Gamma_{1,m}+\Gamma_{2,m}$, where
\[
\Gamma_{1,m} = \sum_{\ell\in L_f} \alpha_{\ell} \prod_{x\in AC}g_m(x)^{\ell_x} \;\;\text{ and }\;\;
\Gamma_{2,m} = \sum_{\ell\in L-L_f} \alpha_{\ell} \prod_{x\in \{0,1\}^k\wedge \ell_x>0}g_m(x)^{\ell_x}.
\]

Next, we will establish a close relationship between $csp_{\Omega}$ and $\Gamma_{1,m}$; more specifically, we intend to prove the following key claim.

\begin{claim}\label{Gamma-1m-a0}
It holds that $(1+\lambda^m B)\Gamma_{1,m}\leq csp_{\Omega}\leq (1+\lambda^m B')\Gamma_{1,m}$ for appropriate numbers $B,B'\in\algebraic$ satisfying $|B|,|B'|\leq a_0$. Therefore, $sgn(csp_{\Omega}) = sgn(\Gamma_{1,m})$ holds.
\end{claim}

\begin{proof}
It is obvious that the second part of the claim follows  from the first part, because $\lambda^m|B| \leq \lambda^ma_0<1$ and similarly $\lambda^m|B'|<1$ by our choice of $m$. Henceforth, we aim at proving the first part. Fix $\ell\in L_f$ arbitrarily. {}From Condition (*), for appropriate selections of $c_{\ell,x}$'s and $d_{\ell,x}$'s in $\{\pm1\}$, we obtain
\begin{equation}\label{eqn:M_0-vs-M_1}
\prod_{x\in AC}(1+\lambda^mc_{\ell,x})^{\ell_{x}} g_m(x)^{\ell_{x}} \leq \prod_{x\in AC}f(x)^{\ell_{x}}\leq \prod_{x\in AC}(1+\lambda^md_{\ell,x})^{\ell_{x}} g_m(x)^{\ell_{x}}.
\end{equation}
Note that, when all elements in $\FF'$ are limited to {\em nonnegative}  constraints, we can always set $c_{\ell,x}=-1$ and $d_{\ell,x}=1$. Eq.(\ref{eqn:M_0-vs-M_1}) leads to upper and lower bounds of $csp_{\Omega}$:
\begin{equation}\label{eqn:alpha-csp-Omega}
\sum_{\ell\in L_f} \alpha_{\ell}\prod_{x\in AC}(1+\lambda^mc_{\ell,x})^{\ell_{x}} g_m(x)^{\ell_{x}} \leq  csp_{\Omega}
\leq \sum_{\ell\in L_f} \alpha_{\ell}\prod_{x\in AC}(1+\lambda^md_{\ell,x})^{\ell_{x}} g_m(x)^{\ell_{x}}.
\end{equation}
Let us further estimate the first and the last terms in Eq.(\ref{eqn:alpha-csp-Omega}). Let us handle the first term.
By considering the binomial expansion of $(1+z)^{n}$, it holds that, for any numbers $n\in\nat^{+}$ and $z\in\real$ satisfying that $-1/n\leq z\leq 2/n$, there exists a number $e\in\{1/2,n\}$ such that $1+nz\leq (1+z)^{n}\leq 1+enz$ (more precisely, if $z\geq0$ then $e=n$; otherwise, $e=1/2$).
Hence, by choosing appropriate numbers $e_{\ell,x} \in \{\pm1/2,\pm\ell_x\}$, we obtain
\begin{eqnarray*}
\prod_{x\in AC}(1+\lambda^m c_{\ell,x})^{\ell_x}g_m(x)^{\ell_x} \geq
\prod_{x\in AC}(1+\lambda^m \ell_{x}e_{\ell,x})g_m(x)^{\ell_x}
= \prod_{x\in AC}(1+\lambda^m \ell_{x}e_{\ell,x})\prod_{x\in AC}g_m(x)^{\ell_x}
\end{eqnarray*}
since $m$ satisfies that $-1<\lambda^m\ell_{x}e_{\ell,x}<1$.

For a further estimation, let us focus on the value $\prod_{x\in AC}(1+\lambda^m z_x)$ for any series $\{z_x\}_{x\in AC}\subseteq [-2^{2k},2^{2k}]_{\integer}$. Since $\prod_{x\in AC}(1+\lambda^m z_x)$ has the form $1+\sum_{i=1}^{|AC|}\sum_{y_1,y_2,\ldots,y_i\in AC}\lambda^{im}z_{y_1}z_{y_2}\cdots z_{y_i}$,  where all indices  $y_1,y_2,\ldots,y_i$ are distinct, if we set $\tilde{B} = \sum_{i=1}^{|AC|}\sum_{y_1,y_2,\ldots,y_i\in AC}\lambda^{(i-1)m}|z_{y_1}z_{y_2}\cdots z_{y_i}|$, then we derive that   $1-\lambda^m \tilde{B}\leq \prod_{x\in AC}(1+\lambda^m z_x)
\leq 1+\lambda^m \tilde{B}$.
Note that $|\lambda^mz_{x}| \leq \lambda^m2^{2k}\leq \lambda^ma_0<1$
for any $x\in AC$ since $z_{x}\in [-2^{2k},2^{2k}]_{\integer}$.
It therefore follows that
\[
\sum_{y_1,\ldots,y_i\in AC}\lambda^{(i-1)m}|z_{y_1}\cdots z_{y_i}|\leq \sum_{y_1\in AC}|z_{y_1}|\sum_{y_2,\ldots,y_i\in AC}1\leq |AC|2^{k}\comb{|AC|}{i}\leq |AC|2^{2k}|AC|!.
\]
We then conclude that $\tilde{B}$ satisfies that
$|\tilde{B}|\leq \sum_{i=1}^{|AC|} |AC|2^{2k}|AC|! \leq |AC|^22^{2k}|AC|! \leq a_0$ since $|AC|\leq 2^k$. {}From this fact, there exists a series $\{B_{\ell}\}_{\ell\in L_f}\subseteq \algebraic$ with $|B_{\ell}|\leq a_0$ such that
\[
\prod_{x\in AC}(1+\lambda^{m}\ell_x e_{\ell,x}) \prod_{x\in AC}g_m(x)^{\ell_x} \geq
(1+\lambda^m B_{\ell})\prod_{x\in AC}g_m(x)^{\ell_x}.
\]
Finally, we choose an appropriate number $B\in\algebraic$ with $|B|\leq a_0$ that satisfies
\begin{equation*}
\sum_{\ell\in L_f}\alpha_{\ell}(1+\lambda^mB_{\ell})\prod_{x\in AC}g_m(x)^{\ell_x} \geq (1 + \lambda^m B)\sum_{\ell\in L_f}\alpha_{\ell}\prod_{x\in AC}g_m(x)^{\ell_x} = (1 + \lambda^m B)\Gamma_{1,m}.
\end{equation*}

Concerning the third term in Eq.(\ref{eqn:alpha-csp-Omega}), a  similar argument used for the first term shows the existence of an algebraic real
number $B'\in\algebraic$ such that $|B'|\leq a_0$ and
\begin{equation*}
\sum_{\ell\in L_f}\alpha_{\ell}\prod_{x\in AC}(1+\lambda^m d_{\ell,x})^{\ell_x}g_m(x)^{\ell_x} \leq
(1+\lambda^m B')\Gamma_{1,m}.
\end{equation*}
By the selection of $B$ and $B'$, they certainly satisfy the claim.
\end{proof}

Next, we will give an upper-bound of $|\Gamma_{2,m}|$. Recall that $b_0= C   \sum_{\ell\in L-L_f}|\alpha_{\ell}|$, where $C= 1+ (2\max|f|)^{|V_2|}$.

\begin{claim}\label{Gamma-2m}
It holds that $|\Gamma_{2,m}|\leq \lambda^m b_0$.
\end{claim}

\begin{proof}
For the time being, we fix a series $\ell\in L-L_f$ and conduct a basic analysis. For this series
$\ell$, there exists an element $x\in\{0,1\}^k$
such that $x\notin AC$ and $\ell_x>0$.
For convenience, we define $D=\{x\in\{0,1\}^k\mid \ell_x>0\}$ and further partition it into two sets: $D_1=\{x\in D\mid x\in AC\}$ and $D_2=\{x\in D\mid x\notin AC\}$. Notice that $D_2$ is nonempty. Since  $\lambda<1$ and $|g_m(x)|\leq \lambda^m$ for all $x\in D_2$, it follows that
\begin{equation*}
\left| \prod_{x\in D_2}g_m(x)^{\ell_x}  \right|
 =  \prod_{x\in D_2}|g_m(x)|^{\ell_x}
\leq \prod_{x\in D_2} \lambda^{m\ell_x}
\leq \lambda^m.
\end{equation*}
Condition (*) implies that $|g_m(x)|\leq |f(x)|/\min\{1+\lambda^mc,1+\lambda^md\}\leq 2|f(x)|$ for any $x\in AC$ because $|\lambda^mc|,|\lambda^md|<1/2$. If $\max|g_m|\geq1$, then we obtain
\[
\left| \prod_{x\in D_1}g_m(x)^{\ell_x} \right| \leq \prod_{x\in D_1} |g_m(x)|^{\ell_x} \leq (\max|g_m|)^{\sum_{x\in D_1} \ell_x}
\leq  (\max|g_m|)^{p_f}      \leq (2\max|f|)^{|V_2|}
\]
since $\sum_{x\in D_1}\ell_x \leq p_f\leq|V_2|$.
When $\max|g_m|<1$, we instead obtain  $\prod_{x\in D_1}|g_m(x)|^{\ell_x}\leq 1$.
Therefore, it holds that
\begin{equation*}
\left| \prod_{x\in D}g_m(x)^{\ell_x}  \right|
= \left| \prod_{x\in D_2}g_m(x)^{\ell_x}  \right|\cdot
\left| \prod_{x\in D_1}g_m(x)^{\ell_x}  \right|
\leq \lambda^m C.
\end{equation*}
The value $|\Gamma_{2,m}|$ is upper-bounded by
\begin{equation*}
|\Gamma_{2,m}| \leq \sum_{\ell\in L-L_f} \left|\alpha_{\ell}\right|  \left|\prod_{x\in D}g_m(x)^{\ell_x} \right|
\leq \lambda^m C \sum_{\ell\in L-L_f}  \left|\alpha_{\ell}\right|
\leq \lambda^m b_0.
\end{equation*}
This completes the proof of the claim.
\end{proof}

To finish the proof of Lemma \ref{constructibility}, we will present a randomized oracle computation that solves $\sharpcsp(f,\GG)$ with
a single query to the oracle $\sharpcsp(\Lambda,\GG)$. First, we want to define a special constant $d_0$ corresponding to $\Omega$. The definition of $d_0$ requires the following well-known lower bound of the absolute values of polynomials in algebraic real numbers.

\begin{lemma}\label{complex-lower-bound}{\rm \cite{Sto74}}\hs{1}
Let $\alpha_1,\ldots,\alpha_m\in\algebraic$ and let $c$ be the degree of $\rational(\alpha_1,\ldots,\alpha_m)/\rational$. There exists a constant $e>0$ that satisfies the following statement for any complex number $\alpha$ of the form $\sum_{k}a_{k}\left(\prod_{i=1}^{m}\alpha_i^{k_i}\right)$, where $k=(k_1,\ldots,k_m)$ ranges over $[N_1]\times\cdots\times[N_m]$, $(N_1,\ldots,N_m)\in\nat^{m}$, and $a_k\in\integer$. If $\alpha\neq0$, then $|\alpha|\geq \left(\sum_{k}|a_k|\right)^{1-c}\prod_{i=1}^{m}e^{-cN_i}$.
\end{lemma}

Following the proof of \cite[Lemma 9.2]{Yam10a} under the assumption that $csp_{\Omega}\neq0$, it is possible to set values of four series $\{N_i\}_{i}$, $\{a_k\}_{k}$, $\{\alpha_i\}_{i}$, and $\{k_i\}_{i}$ appropriately so that Lemma \ref{complex-lower-bound} provides two constants $c,e>0$ for which $|csp_{\Omega}|\geq (\sum_{k}|a_k|)^{1-c}\prod_{i}e^{-cN_i}$. The desired constant $d_0$ is now  defined to be $(\sum_{k}|a_k|)^{1-c}\prod_{i}e^{-cN_i}$. Notice that $d_0$ is an algebraic real number.

Let us describe our randomized approximation algorithm.

\begin{quote}
[Algorithm $\MM$] On input instance $(\Omega,1/\varepsilon)$, set $\delta = \varepsilon/2$ and find  in polynomial time an integer $m\geq1$ satisfying  that $\lambda^ma_0<\min\{1,\delta\}$ and $\lambda^m b_0<\min\{d_0,\delta\}$. Produce another constraint frame $\Omega_m$.
make a query with a query word $(\Omega_{m},1/\delta)$ to the oracle and let $w$ be an answer from the oracle. Notice that $w$ is a {\em random variable} since the oracle is a RAS. Compute $d_0$ defined above. If $|w|< d_0$, then output $0$; otherwise, output $w$.
\end{quote}

We want to prove that the above randomized algorithm $\MM$ approximately
solves $\sharpcsp(f,\GG)$ with high probability. Let us consider two cases separately.

\s

(1) For the first case where $csp_{\Omega}=0$, we need to prove that $M$
outputs $0$ with high probability. Let us evaluate the values $\Gamma_{1,m}$ and $\Gamma_{2,m}$. Obviously, Claim \ref{Gamma-1m-a0} implies $\Gamma_{1,m}=0$. By Claim \ref{Gamma-2m} and the choice of $m$, we derive  $|\Gamma_{2,m}|\leq \lambda^m b_0<d_0$. From $csp_{\Omega_m} = \Gamma_{1,m}+\Gamma_{2,m}$, it follows that $|csp_{\Omega_m}| <d_0$. This means that $\MM$ outputs $0$ with high probability.

\s

(2) Next, we consider the case where $csp_{\Omega}\neq0$.
We consider only the case where $csp_{\Omega}>0$ because the other
case $csp_{\Omega}<0$ can be similarly handled.
The choice of $d_0$ implies that $csp_{\Omega}\geq d_0$.
We then choose a number $\alpha$, not depending on $m$, for which $\alpha (\Gamma_{1,m}+sgn(\alpha)b_0) + b_0 \leq B\Gamma_{1,m}$ and $|\alpha|\leq \max\{a_0,b_0\}$. For this $\alpha$,  it holds by Claim \ref{Gamma-2m} that
\begin{eqnarray*}
(1+\lambda^m\alpha)csp_{\Omega_m} &=&  (1+\lambda^m\alpha)(\Gamma_{1,m}+\Gamma_{2,m}) \\
&\leq& \Gamma_{1,m}+\lambda^m\alpha\Gamma_{1,m} + (1+|\alpha|)\lambda^mb_0 \\
&=&  \Gamma_{1,m}+\lambda^m [\alpha(\Gamma_{1,m} + sgn(\alpha)b_0)+b_0] \\
&\leq& \Gamma_{1,m}+\lambda^m B\Gamma_{1,m} \;\;=\;\;  (1+\lambda^mB)\Gamma_{1,m}.
 \end{eqnarray*}
Similarly, we choose $\alpha'$ with $|\alpha'|\leq \max\{a_0,b_0\}$ such that $(1+\lambda^mB)\Gamma_{1,m}\leq (1+\lambda^m\alpha')(\Gamma_{1,m}+\Gamma_{2,m}) = (1+\lambda^m\alpha')csp_{\Omega_m}$.

For simplicity, let $\gamma = \max\{|\alpha|,|\alpha'|\}$. Note that $\delta\geq \lambda^m\gamma$. Since $\lambda^m\gamma<1$, it holds that $\log_2(1+\lambda^m\gamma)\leq \lambda^m\gamma\leq \delta$.
Thus, we conclude that $1+\lambda^m\alpha'\leq 2^{\log_2(1+\lambda^m\gamma)}\leq 2^{\delta}$. Moreover, since $\log_2(1-\lambda^m\gamma)\geq -\lambda^m\gamma$, it follows that  $1+\lambda^m\alpha \geq 2^{\log_2(1-\lambda^m\gamma)}\geq 2^{-\delta}$.
In conclusion, it holds that
$2^{-\delta} csp_{\Omega} \leq csp_{\Omega_m} \leq 2^{\delta} csp_{\Omega}$.  {}From this follows $csp_{\Omega_m}>0$.

If $w$ is any oracle answer, then it must satisfy that $2^{-\delta}csp_{\Omega_m}\leq w \leq 2^{\delta}csp_{\Omega_m}$ because $csp_{\Omega_m}>0$. Therefore, we derive that $w\leq 2^{\delta}csp_{\Omega_m}\leq 2^{2\delta}csp_{\Omega}$ and $w\geq 2^{-\delta}csp_{\Omega_m}\geq 2^{-2\delta}csp_{\Omega}$. Since $\varepsilon = 2\delta$, $\MM$ outputs a $2^{\varepsilon}$-approximate solution using any $2^{\delta}$-approximate solution for $(\Omega_m,1/\delta)$ as an oracle answer.

\s

This completes the proof of Lemma \ref{constructibility}

%%%%%%%%%%%%%%%%%%%%%%%%%%
%%%%%%%%%%%%%%%%%%%%%%%%%%
\let\oldbibliography\thebibliography
\renewcommand{\thebibliography}[1]{%
  \oldbibliography{#1}%
  \setlength{\itemsep}{0pt}%
}
\bibliographystyle{plain}

%%%%%%%%%%%%%%%%%%%%%%%%
%%%%%%%%%%%%%%%%%%%%%%%%
\end{document}